\documentclass[runningheads]{llncs}
\usepackage[T1]{fontenc}
\usepackage[ruled,vlined]{algorithm2e}
\usepackage{graphicx}
\usepackage{amsfonts, amssymb, amsmath}
\usepackage{subcaption}
\usepackage{xcolor}

\DeclareMathOperator*{\argmin}{arg\,min}

\begin{document}
\title{Facility Location Problem under Local Differential Privacy without Super-set Assumption}
\titlerunning{Facility Location Problem under Local Differential Privacy}
%
\author{Kevin Pfisterer\inst{1}\orcidID{0009-0009-2631-6533} \and \\
Quentin Hillebrand\inst{2}\orcidID{0000-0002-7747-4998}
\and \\
Vorapong Suppakitpaisarn\inst{2}\orcidID{0000-0002-7020-395X} 
}
\authorrunning{Pfisterer K. et al.}
%
\institute{Technical University of Munich, Germany\\
\email{kevin.pfisterer@tum.de}\and
The University of Tokyo, Japan\\
\email{quentin-hillebrand@g.ecc.u-tokyo.ac.jp,vorapong@is.s.u-tokyo.ac.jp}
}
\maketitle              

\begin{abstract}
In this paper, we introduce an adaptation of the facility location problem and analyze it within the framework of local differential privacy (LDP). Under this model, we ensure the privacy of client presence at specific locations. 
When $n$ is the number of points, Gupta et al.~\cite{gupta} established a lower bound of \(\Omega(\sqrt{n})\) on the approximation ratio for any differentially private algorithm applied to the original facility location problem. As a result, subsequent works have adopted the superset assumption, which may, however, compromise user privacy.
We show that this lower bound does not apply to our adaptation by presenting an LDP algorithm that achieves a constant approximation ratio with a relatively small additive factor. Additionally, we provide experimental results demonstrating that our algorithm outperforms the straightforward approach on both synthetically generated and real-world datasets.

\keywords{Privacy in location-based services \and Local differential privacy \and Facility location \and Approximation algorithms}
\end{abstract}
\section{Introduction}

The facility location problem is a well studied problem in combinatorial optimization and operations research. 
Given a set of locations, costs for opening facilities, and a metric for the distance between locations, the goal is to open a set of facilities and connect clients to the facilities with minimal costs.
The problem is NP-hard and therefore research focused on developing heuristic approaches and approximation algorithms since the 1960's \cite{stollsteimer,kuehn}. 
The problem finds application in several fields such as data mining, bioinformatics and machine learning \cite{esencayi}. 
It can be formalized as follows.
\begin{definition}[Facility location problem \cite{cohen}]
    Given the tuple $(V, d, \vec{f}, \vec{b})$ where $V$ is the set of locations with $|V| = n$, $(V,d)$ is a metric, $\vec{f} \in \mathbb{R}_{\ge 0}^{n}$ indicates the facility costs for every location $v \in V$, and $\vec{b} \in \{0,1\}^n$ indicates if a client is present at location $v \in V$. 
    The objective of the facility location problem is to find a set of locations $S \subseteq V$ which minimizes:
    \begin{equation}\label{eq:flp_cost}
        cost_d(S,\vec{b}) = \sum_{s \in S} f_s + \sum_{v \in V} b_vd(v,S)
    \end{equation}
    with $d(v,S) = \min\limits_{s \in S} d(v,s)$. 
    The first part of Equation \ref{eq:flp_cost} is called facility costs while the latter one is called connection costs.
\end{definition}

With the development of differential privacy (DP) by Dwork et al. \cite{dwork}, recent research applied DP to the facility location problem to ensure privacy for clients. 
The idea of differential privacy is to ensure that the inclusion or exclusion of a single data point does not significantly affect the solution, thereby protecting individual privacy. The concept measures a system's privacy leakage using a parameter called the privacy budget, denoted as \(\varepsilon\). An algorithm is referred to as an \(\varepsilon\)-DP mechanism if it has a privacy budget of \(\varepsilon\).

Differential private algorithms deploy basic privacy mechanisms such as the exponential mechanism \cite{mcsherry2007mechanism} or the Laplace mechanism \cite{dwork} depending on the domain of the solution. 
Both mechanisms introduce a degree of uncertainty about the correctness of the solution. The Laplace mechanism does this by adding a noise drawn from the Laplace distribution to the solution while the exponential mechanism assigns output probabilities to categorical solution proportional to their utility.

DP requires a trusted curator that runs the algorithm and has access to the private information.
This introduces a single point of failure and makes the curator prone to malicious attacks or human errors. 
Local differential privacy (LDP) \cite{kasiviswanathan} solves this issue by restricting the access of private data to the clients themselves. 
LDP algorithms are split into two parts.
The first one runs locally and applies a privacy mechanisms to the private data to generate a noisy version. 
The noisy data is then sent to an aggregation server that computes a solution to the problem based on the noisy input.

There are works applying DP to the facility location problem.
Gupta et al. \cite{gupta} were the first to apply DP to the uniform facility location problem. In this setting all locations have the same facility costs. 
They showed that any differentially private algorithm for the uniform facility location problem has an approximation ratio of $\Omega(\sqrt{n})$ when $n$ is the number of points.  

This lead to the introduction of the super-set output assumption \cite{gupta,cohen}. Here, the output of the algorithm for the facility location problem is a set $R \subseteq V$ of potential facilities. 
Every client is then connected according to a predefined connection rule to one of the facilities in $R$. 
The actual set of opened facilities $S$ is then the set of locations in $R$ that have at least one connected client. 
Under this super-set output setting, they provide an $\varepsilon$-DP algorithm with expected costs of at most 
\[
    OPT \cdot O(\log n \log \Delta) \cdot \frac{\log \Delta}{\varepsilon}\log{\frac{n \log^2{\Delta}}{\varepsilon}}
\]
where $OPT$ is the optimal value and $\Delta = \max\limits_{u,v\in V}{d(u,v)}$.
It uses results from Fakcharoenphol et al. \cite{fakcharoenphol} to approximate any arbitrary metric by a distribution over hierarchically well-separated trees (HST) with distortion $O(\log n)$. 
The algorithm then takes an HST as input and outputs a super-set of facilities.

In \cite{esencayi}, Esencayi et al. improve upon the results of Gupta et al. for the facility location problem under the $\varepsilon$-DP model with super-set output setting.
They showed that there exists an algorithm with an approximation ratio of $O(\frac{1}{\varepsilon})$ under the HST metric and therefore an algorithm with approximation ratio of $O(\frac{\log n}{\varepsilon})$ for any arbitrary metric.
Furthermore, they proved that the approximation ratio of any $\varepsilon$-DP algorithm is lower bounded by $\Omega(\frac{1}{\sqrt{\varepsilon}})$ even under an HST metric and the super-set output setting.

Since these approaches first convert the arbitrary metric into an HST-metric, they cannot perform better than $O(\log n)$.
This leaves the question upon whether the approximation ratio under the super-set output setting needs to grow with $n$.
Recently, Pasin Manurangsi in \cite{manurangsi2025} answered this question positively by providing a lower bound of $\tilde{\Omega}(min\{\log n,\sqrt{\frac{\log n}{\varepsilon}}\})$ on the expected approximation ratio of any $\varepsilon$-DP algorithm. Since $\varepsilon$-LDP is more restrictive than $\varepsilon$-DP, the lower bound also applies to $\varepsilon$-LDP algorithms. 
This further validates our approach of constructing a problem that finds real-world application and where the super-set output setting is no longer necessary.

Cohen-Added et al. \cite{cohen} were the first to study the facility location problem under local differential privacy (LDP) in 2022.
They provide an $\varepsilon$-LDP algorithm that achieves an $O(n^{1/4}/\varepsilon^2)$ approximation ratio under the HST metric for the non-uniform facility location problem. The algorithm applies the randomized response mechanism, a special variant of the exponential mechanism, on the location side to private data. It uses an HST together with the noisy data from every location to output a super-set of facilities. The clients are then connected based on a lowest common ancestor rule to facilities.
Furthermore, they proved a lower bound of $O(n^{1/4}/\sqrt{\varepsilon})$ for the approximation ratio of any non-interactive $\varepsilon$-LDP algorithm. 

\subsection{Our Contribution} 
Although most works on differentially private facility location algorithms assume the super-set assumption, this compromises user privacy. Revealing only facilities that connect to at least one user discloses information about users. This concern motivates us to explore a practical setting for the facility location problem, where we can design an LDP algorithm with a constant approximation ratio and a relatively small additive factor. We show experimentally that the proposed algorithm outperforms the straightforward solution on synthetically generated data and a dataset based on real world data from the city Chiang Mai, Thailand.

We introduce the Facility Location Problem with Linear Facility
Costs.  
In this setting, the objective is to assign a capacity to each opened facility, enabling it to serve up to its designated limit of clients. Additionally, facility costs scale linearly with capacity, meaning that the higher a facility's capacity, the greater the cost to establish it.  
We demonstrate that this problem is no longer NP-hard and can be optimally solved in polynomial time. Also, the negative result by Gupta et al. no longer applies in our setting, allowing us to design efficient \(\varepsilon\)-LDP algorithms without relying on the super-set output setting.

To obtain an $\varepsilon$-LDP algorithm for this setting, one can deploy the Laplace mechanism locally and then compute the capacities based on the noisy input. 
When the algorithm operates with a failure probability \(\alpha\), we show that this simple approach achieves an expected upper bound on the cost of \(\left(1 + \frac{2}{\varepsilon} \ln{\frac{2n}{\alpha}}\right) OPT\). The resulting approximation ratio is \(O(\log n)\).

We propose an algorithm to further improve that trivial algorithm. For this algorithm, we assume that, for an input where every location $v \in V$ has at least $\gamma^2\ln^2{n}$ locations with a distance of at most $\delta$ away, we propose an $\varepsilon$-LDP algorithm that bounds the total costs by $(1 + \frac{2}{\varepsilon}\ln{\frac{2n}{\alpha}}\frac{1}{\gamma\ln{n}})OPT + \delta n(\frac{2}{\varepsilon}\ln{\frac{2n}{\alpha}}\frac{1}{\gamma\ln{n}} + 4b_{avg})$, when $b_{avg}$ is the average number of clients per location. We observe that the approximation ratio \(\left(1 + \frac{2}{\varepsilon} \ln{\frac{2n}{\alpha}} \frac{1}{\gamma \ln n} \right) = O(1)\), and the additive term \(\delta n(\frac{2}{\varepsilon}\ln{\frac{2n}{\alpha}}\frac{1}{\gamma\ln{n}} + 4b_{avg})\) remains small compared to overall cost when \(\delta\) is sufficiently small.

Finally, we demonstrate that our \(\varepsilon\)-LDP algorithm outperforms the straightforward approach on both synthetically generated datasets and a real-world instance based on data from Chiang Mai, Thailand. We use the Matérn cluster point process to generate clustered instances that simulate cities with densely populated neighborhoods. Furthermore, we use the Poisson point process to generate instances with uniformly at random distributed locations.
We show that, for both generations processes with varying generation parameters and for the real-world instance, our theorems provide input parameters for our $\varepsilon$-LDP algorithm resulting in lower cost solutions compared to the straightforward approach.

\section{Local Differential Privacy}
The following definition formally introduces $L_1$ local differential privacy (LDP). 
\begin{definition}[$L_1$-LDP \cite{kasiviswanathan}]
Let \(\varepsilon > 0\). A randomized query \(\mathcal{R}\) satisfies $L_1$ \(\varepsilon\)-local differential privacy (\(\varepsilon\)-LDP) if, for any possible local inputs \(b,b' \in \mathbb{Z}\) such that $|b - b'| = 1$, and any possible outcome set \(S\), 
    $$\Pr[\mathcal{R}(b) \in S] \leq e^{\varepsilon} \Pr[\mathcal{R}(b') \in S].$$
An algorithm \(\mathcal{A}\) is said to be $L_1$ \(\varepsilon\)-LDP if, for any local node, and any sequence of queries \(\mathcal{R}_1, \dots, \mathcal{R}_\kappa\) posed the node, where each query \(\mathcal{R}_j\) satisfies \(\varepsilon_j\)-local differential privacy (for \(1 \leq j \leq \kappa\)), the total privacy loss is bounded by \(\varepsilon_1 + \dots + \varepsilon_\kappa \leq \varepsilon\).    
The algorithm is referred to as non-interactive if the number of queries to each user is one, and those queries are independent to each other.
\end{definition}
In the following section, we explain our rationale for employing the notion of \( L_1 \)-LDP in this paper. 

\begin{definition}[$L_1$-Sensitivity \cite{dwork}] The \( L_1 \)-sensitivity of a function \( f: \mathbb{Z} \to \mathbb{R} \) is the smallest value \( S(f) \) such that, for all \( b, b' \in \mathbb{Z} \) such that $|b - b'| = 1$, the following holds:  
\[
\|f(b) - f(b')\|_1 \leq S(f)
\]
\end{definition}

\begin{definition}[$L_1$ Local Laplacian Query \cite{hillebrand2023unbiased}]
For any function \( f: \mathbb{Z} \to \mathbb{R} \) and input \( b \in \mathbb{Z} \), the following mechanism, known as the Laplace mechanism, is defined as:  
\[
LM_f(b) = f(b) + Y
\]
where \( Y \) is drawn from \( \text{Lap}(S(f)/\varepsilon) \). The $L_1$ local Laplace mechanism satisfies $L_1$ \(\varepsilon\)-LDP.
\end{definition}

Since this paper exclusively employs the $L_1$ version of $\varepsilon$-LDP, we omit the $L_1$ prefix. In other words, when we mention $\varepsilon$-LDP, sensitivity, and the local Laplacian query, we are referring to the $L_1$ $\varepsilon$-LDP, $L_1$-sensitivity, and the $L_1$ local Laplacian query, respectively.

\section{Our Setting: Facility Location with Linear Facility Costs (FL-Linear)}
In this section, we formally present our adaptation of the traditional facility location problem, referred to as Facility Location with Linear Facility Costs (FL-Linear). We later demonstrate that this problem can be optimally solved in polynomial time.

\subsection{Problem Statement}

\begin{definition}[FL-Linear]\label{def:problem}
    The Facility Location with Linear Facility Costs is defined by the input tuple $(V, d, \vec{f}, \vec{b})$ where $V$ defines the set of locations with $|V| = n$, $(V, d)$ is a metric, $\vec{f}$ indicates the facility costs for every location $v \in V$, and $\vec{b} = [b_v]_{v \in V} \in \mathbb{N}^n_{>0}$ indicates the number of clients present at location $v \in V$.
    The goal is to find a set of locations $S \subseteq V$ with capacities $(k_s)_{s\in S}$ and a connection function $h: V \rightarrow S$ that minimize the costs
    \begin{equation}
        cost_{d}(S, h, \vec{b}) = \sum_{s \in S} k_sf_s + \sum_{v \in V} b_vd(v, h(v))
    \end{equation}
    Furthermore, every facility must be able to serve all of its connected clients. Let $L_v = \{u\in V : h(u) = v\}$ denote the set of connected locations to a facility $v \in S$. Then the following equation must hold for all $v \in S$:
    \begin{equation}
        \sum_{u \in L_v} b_u \le k_v
    \end{equation}
\end{definition}

In contrast to the original problem, every facility $v \in S$ now has a capacity $k_v$ assigned to it. Furthermore, the facility to which a location is connected to is now given as an output by the connection function $h$. 
Previous research on the original problem defined a connection rule (e.g. connect location to closest open facility) and only outputted a super-set of opened facilities.

Since a location appears in the dataset only if at least one person resides there, we assume that every location has at least one client present, i.e., \( b_v > 0 \). In Section \ref{chapter:other_results}, we relax this assumption and derive a bound based on a client-location ratio and a Bernoulli distributed client presence.

The FL-Linear problem finds application in real-world scenarios such as in the setting of placing evacuation shelters. The costs for providing space, food, water and, equipment like first aid kits depends on the amount of people it is designed to shelter. Therefore, it is reasonable to determine the facility cost based on the number of individuals relocating to the facility.

\paragraph{Privacy Assumption} We adopt the assumption from previous works that the set of locations, the distances between them, and the facility costs at each location are publicly available information. The only data kept private is the number of individuals at each location, denoted as \(\mathbf{b}\). This assumption is motivated by the problem of placing facilities for evacuation \cite{real_world_data}. In such scenarios, the distances between buildings and the costs of constructing evacuation facilities at each site are known. However, the presence of an individual at a specific location is sensitive information. Revealing the number of individuals at each location could risk disclosing this sensitive data. 

We consider the value \( b_u \in \mathbb{Z} \) as local data. To safeguard an individual's location, we seek a mechanism that ensures others cannot distinguish between \( b_u \) and \( b'_u \) when \( |b_u - b'_u| = 1 \), which corresponds to the presence or absence of the person at node \( u \). For this reason, we adopt the privacy notion of \( L_1 \) \( \varepsilon \)-LDP in this work.

\subsection{Non-Private Optimal Algorithm}
By introducing linear facility costs based on the capacity, the problem becomes easier and we can find a polynomial time algorithm to solve it. 
For a fixed location $v$ the costs it induces is independent on the connection of other locations. We say that $v \in V$ is connected to $u \in V$ if $h(v) = u$.
Finding the optimal facility $u$, a location $v$ should be connected to, is now a local decision: 
\[
    h(v) = \argmin_{u \in V} (f_u + d(u,v))
\]

For each location, we determine the optimal location to which it is connected. We define the set of locations that have at least one connection as  
\[
M = \{v \in V : \exists u \in V \text{ such that } h(u) = v\} = \{v \in V : h(v) = v\}.
\]
For any marked location \( v \in M \), let \( L_v = \{u \in V : h(u) = v\} \) denote the set of locations assigned to \( v \). The facility at \( v \) is then opened with a capacity corresponding to the total demand of the clients connected to it:  
\[
k_v = \sum_{u \in L_v} b_u.
\]
Algorithm \ref{alg:opt} summarizes the described methods.
In the base algorithm, no privacy mechanism is applied. Consequently, the algorithm has full knowledge of the exact client demands \( b_u \) and can set facility capacities accordingly. As a result, every location is optimally connected, leading to a globally optimal solution. The computation can be performed in polynomial time.

\begin{algorithm}
    \DontPrintSemicolon
    \KwData{$V$, $d$, $\vec{f}$, $\vec{b}$}
    \KwResult{connection function $h$, capacities $k$}
    \Begin{
        \For{$v \in V$}{
            $h(v) \longleftarrow \argmin\limits_{u \in V} f_u + d(u,v)$\;
        }
        $M \longleftarrow \{v\in V : h(v) = v\}$\;
        \For{$v \in M$}{
            $L_v\longleftarrow \{u \in V:h(u)=v\}$\;
            $k_v \longleftarrow \sum_{u \in L_v} b_u$\;
        }
        \Return $h$, $k$
    }
    \caption{non-private optimal algorithm}\label{alg:opt}
\end{algorithm}

\section{Straightforward Algorithm:  Laplace Mechanism with Margin}\label{chapter:naive_algo}
To keep the amount of clients $b_v$ of a location $v$ private, 
 we introduce a local differential private algorithm that uses the Laplace mechanism to ensure privacy and opens facilities based on a noisy number of clients.
Every location $v \in V$ adds a Laplacian noise parameterized by the privacy budget $\varepsilon$ to their private number of clients $b_v$ to generate a noisy variant $b'_v$.
Afterwards, it sends $b'_v$ to the aggregation server. On the server side the optimal assignments and capacities are computed based on $\mathbf{b}' = [b'_v]_{v \in V}$.

\begin{algorithm}
    \DontPrintSemicolon
    \KwData{$V$, $d$, $\vec{f}$, $\vec{b}$, privacy budget $\varepsilon$, failure probability $\alpha$}
    \KwResult{connection function $h'$, capacities $k'$}
    \Begin{
        \textbf{Location Side:}\;
        \For{$v \in V$}{
            $b'_v \longleftarrow b_v + Lap(1/\varepsilon)$\;
            Send $b'_v$ to server\;
        }
        \textbf{Server Side:}\;
        \For{$v \in V$}{
            $h'(v) \longleftarrow \argmin\limits_{u \in V} f_u + d(u,v)$\;
        }
        $M' \longleftarrow \{v\in V : h'(v) = v\}$\;
        \For{$v \in M'$}{
            $L_v'\longleftarrow \{u \in V:h'(u)=v\}$\;
            $N'_v \longleftarrow \sum_{u \in L'_v} b'_u$\;
            $k'_v \longleftarrow N'_v + \frac{2} {\varepsilon}\sqrt{|L'_v|}\ln{\frac{2n}{\alpha}}$\;
        }
        \Return $h'$, $k'$
    }
    \caption{$\varepsilon$-LDP algorithm with Laplacian mechanism and margin}\label{alg:no_reconn}
\end{algorithm}

Algorithm \ref{alg:no_reconn} mirrors the steps of the non-private algorithm by first computing the optimal connections between locations. We leverage the key observation that for any location \(v\), 
\[
h(v) = \argmin_{u \in V} \big(b_vf_u + b_vd(u,v)\big) = \argmin_{u \in V} \big(f_u + d(u,v)\big) = h'(v),
\]
which means we do not need direct access to the private value \(b_v\) to calculate $h'(v)$. The main distinction between the two approaches arises during the capacity computation. In the optimal assignment, the capacity can be set exactly to the number of connected clients. However, in the private setting, the noisy estimate \(b'_v\) may underestimate the true client count \(b_v\). Without an additional margin, this underestimation would lead to an immediate failure at location \(v\). Consequently, we add a margin of 
$\frac{2}{\varepsilon}\sqrt{|L'_v|}\ln{\frac{2n}{\alpha}}
$ to the noisy count of connected clients \(N'_v\) to ensure that the total failure probability does not exceed \(\alpha\).

\subsection{Analysis}
We first show that Algorithm \ref{alg:no_reconn} ensures privacy for the presence of individuals.
\begin{theorem}\label{theorem:trivial_privacy}
    Algorithm \ref{alg:no_reconn} satisfies $\varepsilon$-LDP.
\end{theorem}
\begin{proof}
    The presence or abscence of an individual at location $v$ changes $b_v$ by at most 1. Hence, the $L_1$-sensitivity is 1. Algorithm \ref{alg:no_reconn} locally adds Laplacian noise drawn from $Lap(1/\varepsilon)$ to $b_v$ and therefore satisfies $\varepsilon$-LDP. 
\end{proof}
In the following we show that the output by Algorithm \ref{alg:no_reconn} satisfies the capacity constraint with a probability of $1-\alpha$ and if they are satisfied the expected costs are bounded by $(1+\frac{2}{\varepsilon}\ln{\frac{2n}{\alpha}})OPT$.
Let $E_v$ denote the event that a failure occurs at location $v \in M'$. This means more clients are connected to $v$ than it has capacity: $\sum_{u \in L'_v} b_u > k'_v$. 
For a failure to occur, the Laplacian noise added to the clients in $L'_v$ must be larger than the margin added to $v$, i.e.
\[
    \left|\sum_{u \in L'_v}Lap(1/\varepsilon)\right| > \frac{2}{\varepsilon}\sqrt{|L'_v|}\ln{\frac{2n}{\alpha}} \label{eqn:noise}
\]
We first bound $\Pr[E_v]$, and then apply union bound to derive a bound for the total failure probability. For the first part, we use the following results.

\begin{theorem}[Xian et al. \cite{xian2024differentially}]\label{theorem:sum}
    Let $X_1, ..., X_k \sim Lap(b)$ be independent, then for $t \ge 2b\sqrt{k}\ln{\frac{2k}{\beta}}$,
    \begin{equation}
    \Pr\left[\left|\sum_{i=1}^{k} X_i\right| \le t\right] \ge 1-\beta
    \end{equation}
\end{theorem}
From Theorem \ref{theorem:sum}, a bound on $\Pr[E_v]$ follows by setting $t$ according to the margin.

\begin{theorem}\label{theorem:no_reconn_failure}
    The total failure probability of Algorithm \ref{alg:no_reconn} is bounded by $\alpha$.
\end{theorem}
\begin{proof}
Let $v \in M'$, then it will be opened with a capacity of 
\[
    k'_v = N'_v + \frac{2}{\varepsilon}\sqrt{|L'_v|}\ln{\frac{2n}{\alpha}}
\]
A failure occurs if the capacity is smaller than the actual number of connected clients. 
\[
    \Pr\left[N_v > N'_v + \frac{2}{\varepsilon}\sqrt{|L'_v|}\ln{\frac{2n}{\alpha}}\right]
    \le \Pr\left[\left|\sum_{u \in L'_v}Lap(1/\varepsilon)\right| > \frac{2}{\varepsilon}\sqrt{|L'_v|}\ln{\frac{2n}{\alpha}}\right]
\]
Given the total failure probability $\alpha$, by Theorem \ref{theorem:sum} with $k = |L'_v|$ and $\beta_v = \frac{|L'_v|}{n} \alpha$ it follows that for $t \ge \frac{2}{\varepsilon}\sqrt{|L'_v|}\ln{\frac{2n}{\alpha}}$, 
\[
    \Pr[E_v] \le \frac{|L'_v|}{n} \alpha.
\]
By applying union bound we get a bound on the total failure probability,
\[
    \sum_{v\in M'}\Pr[E_v] \le \frac{\alpha}{n}\sum_{v \in M'} |L'_v| = \alpha.
\]

\end{proof}

\begin{theorem}\label{theorem:no_reconn_costs}
    Assuming no failures occur, the expected cost of the solution produced by Algorithm \ref{alg:no_reconn} is at most 
\[
\left(1 + \frac{2}{\varepsilon}\ln\frac{2n}{\alpha}\right)OPT.
\]
\end{theorem}
\begin{proof} Recall that the optimal cost $OPT$ can be divided into the connection cost, denoted by $OPT_{conn}$ and the facility cost $OPT_{fac}$, i.e. $OPT = OPT_{conn} + OPT_{fac}$. 
    Computing the connection function $h': V\rightarrow V$ does not involve private information.
    Since no failure occurs, $h'$ is a valid solution and will be the same as the connection function $h$ from the optimal non-private algorithm. 
    Therefore, the connection costs are the same as in the optimal solution:
    \[
        \sum_{v\in V}b_vd(v,h'(v)) = OPT_{conn}
    \]
    Since, the connection function is the same, also the set of marked locations $M'$ is the same as $M$. Algorithm \ref{alg:no_reconn} opens $v \in M'$ with a capacity of $k'_v = N'_v + \frac{2}{\varepsilon}\sqrt{|L'_v|}\ln{\frac{2n}{\alpha}}$. This gives expected facility costs of
    \[
        E\left[\sum_{v \in M'} \left(N'_v + \frac{2}{\varepsilon}\sqrt{|L'_v|}\ln{\frac{2n}{\alpha}}\right)f_v\right] = OPT_{fac} + \frac{2}{\varepsilon}\ln{\frac{2n}{\alpha}}\sum_{v \in M'}\sqrt{|L'_v|}f_v
    \]
    Finally, we obtain that our costs are
    \[
        OPT_{conn} + OPT_{fac} + \frac{2}{\varepsilon}\ln{\frac{2n}{\alpha}}\sum_{v \in M'}\sqrt{|L'_v|}f_v \le \left(1 + \frac{2}{\varepsilon}\ln{\frac{2n}{\alpha}}\right)OPT.
    \]
    $\sum_{v \in M'}\sqrt{|L'_v|}f_v$ can be bounded with $OPT$ since every location has at least one client and therefore a facility is opened with a capacity of at least $|L'_v|$ in the optimal case. 
\end{proof}
In this paper, we do not incorporate penalty costs for facility failures. Our analysis can be extended to account for such failures by introducing an appropriate penalty function and using Theorem \ref{theorem:no_reconn_failure} to bound the expected total costs.

\section{Our Algorithm: $\varepsilon$-LDP Algorithm with Reconnection}\label{chapter:our_algo}
The margin added to every marked node $v \in M'$ in Algorithm \ref{alg:no_reconn} depends on the number of connected locations $|L'_v|$. 
This leaves the question whether the total margin added can be decreased while keeping the upper bound on the failure probability.
In the last step of the analysis of Algorithm \ref{alg:no_reconn}, we bounded $\sum_{v \in M'}\sqrt{|L'_v|}f_v$ with $OPT$. 
In this section, we describe how by merging close facilities, we can create less facilities with higher capacities and therefore find a better bound for $\sum_{v \in M'}\sqrt{|L'_v|}f_v$. 
This decreases the multiplicative error from $\mathcal{O}(\log{n})$ to constant under an additional assumption about the distribution of the locations while introducing an additive error.

\paragraph{Additional Assumption} In the following we assume that no location is isolated from all other locations. Formally, for a given $\delta > 0$, let $B(v, \delta)$ denote the ball of radius $\delta$ centered on the location $v \in V$. A location $u \in V$ is contained in the ball $B(v, \delta)$ if $d(u,v) \le \delta$. 
We assume that for $\delta > 0, \gamma \ge 1$ and every location $v \in V$: \begin{equation}\label{eq:local}
    |B(v, \delta)| \ge \gamma^2 \ln^2{n}
\end{equation}
With \(\gamma = 1\), \(n = 10,000,000\), and \(\delta = 1 \text{ km}\), our assumption implies that each household has at least 259.79 other households within a 1 km radius. 
Given that the average Japanese household size in 2022 was 2.25, this corresponds to approximately 585 inhabitants within the same area, resulting in a population density of 186 inhabitants per square kilometer. This confirms that our assumption remains valid even in small village settings. For cases focusing on densely populated areas, such as Tokyo, which has a population density of 6,300 people per square kilometer \cite{jap_pop}, our assumption is even more strongly supported.

\subsection{Description of $\varepsilon$-LDP Algorithm with Reconnection}
The $\varepsilon$-LDP algorithm with reconnection follows similar steps as Algorithm \ref{alg:no_reconn}. 
On the location side the Laplace mechanism is used to ensure $\varepsilon$-LDP. 
Then it computes the optimal assignment $\hat{h}$ and the set of marked locations $\hat{M}$ without using private data.

Instead of immediately opening every marked location with a margin as done in Algorithm \ref{alg:no_reconn}, we select a maximal set of marked locations that are pairwise at least a distance of \(2\delta\) apart. This ensures that no two facilities opened are within \(2\delta\) of each other. To compute this set, we first construct the graph
$G = \Bigl(\hat{M},\ \Bigl\{\{u,v\} : u,v \in \hat{M},\ d(u,v) \le 2\delta\Bigr\}\Bigr)$.
Then, we run a greedy algorithm to determine a maximal independent set \(I\) in \(G\). In this process, the algorithm repeatedly selects the node with the lowest facility value for which none of its neighbors has already been chosen, ensuring that every pair of nodes in \(I\) is separated by at least \(2\delta\).

Before computing the capacities for the facilities in \(I\), we update the connection function \(\hat{h}\). For every \(v \in I\), we connect all nodes in \(B(v,\delta)\) to \(v\). Moreover, all locations not covered by any ball are connected to the optimal location restricted to the set \(I\). Finally, similar to Algorithm \ref{alg:no_reconn}, we add a margin of 
$\frac{2}{\varepsilon}\sqrt{|\hat{L}_v|}\ln\frac{2n}{\alpha}$
to the capacity, on top of the number of connected clients \(\hat{N}_v\).

\begin{algorithm}[h]
    \DontPrintSemicolon
    \KwData{$V$, $d$, $\vec{f}$, $\vec{b}$, privacy budget $\varepsilon$, failure probability $\alpha$, $\delta$}
    \KwResult{connection function $\hat{h}$, capacities $\hat{k}$}
    \Begin{
        \textbf{Location Side:}\;
        \For{$v \in V$}{
            $b'_v \longleftarrow b_v + Lap(1/\varepsilon)$\;
            Send $b'_v$ to server\;
        }
        \textbf{Server Side:}\;
        \For{$v \in V$}{
            $\hat{h}(v) \longleftarrow \argmin_{u \in V} f_u + d(u,v)$\;
        }
        $\hat{M} \longleftarrow \{v\in V : \hat{h}(v) = v\}$\;
        $G \longleftarrow (\hat{M}, \{\{u, v\} : u,v\in \hat{M}, d(u,v)\le 2\delta\})$\;
        $I \longleftarrow$ maximal independent set of $G$\;
        \For{$v \in I$}{
            $\hat{h}(u) \longleftarrow v$ for all $u \in B(v, \delta)$\;
        }
        \For{$u \in V$ s.t. $\nexists v \in I$ with $u \in B(v, \delta)$}{
            $\hat{h}(u) \longleftarrow \argmin_{v \in I} f_v + d(u,v)$\;
        }
        
        \For{$v \in I$}{
            $\hat{L}_v\longleftarrow \{u \in V:\hat{h}(u)=v\}$\;
            $\hat{N}_v \longleftarrow \sum_{u \in \hat{L}_v} b'_u$\;
            $\hat{k}_v \longleftarrow \hat{N}_v + \frac{2} {\varepsilon}\sqrt{|\hat{L}_v|}\ln{\frac{2n}{\alpha}}$\;
        }
        \Return $\hat{h}, \hat{k}$
    }
    \caption{$\varepsilon$-LDP algorithm with reconnection}\label{alg:reconn}
\end{algorithm}

\subsection{Analysis}
The reconnection algorithm uses the same Laplace privacy mechanism as Algorithm \ref{alg:no_reconn}. Therefore, the proof of $\varepsilon$-LDP is straightforward.
\begin{theorem}
    Algorithm \ref{alg:reconn} satisfies $\varepsilon$-LDP.
\end{theorem}
\begin{proof}
    Similar to Theorem \ref{theorem:trivial_privacy}, the absence or presence can change $b_v$ by at most 1, resulting in a sensitivity of 1. 
    Algorithm \ref{alg:reconn} applies the local Laplace mechanism by adding Laplacian noise with parameter $1/\varepsilon$. 
\end{proof}

In the following, we analyze the costs of the output of $\varepsilon$-LDP algorithm with reconnection (Algorithm \ref{alg:reconn}). 

\begin{lemma}
    \(|\hat{L}_v| \ge \gamma^2 \ln^2 n\) for every \(v \in I\). \label{lem:sizeLv}
\end{lemma}
\begin{proof}
    Because the set \(I\) has the property that the balls \(B(u,\delta)\) for all \(u \in I\) do not overlap, reconnecting all nodes in \(B(v,\delta)\) to \(v\) together with the assumption that \(|B(v,\delta)| \ge \gamma^2 \ln^2 n\) establishes the lemma.
\end{proof}

We bound the additional costs that occur due to the reconnection and then bound the costs of the additional capacity used to open facilities compared to the optimal assignment.

\begin{lemma}
    With $b_{avg} = \frac{1}{n}\sum_{v \in V} b_v$ the extra reconnection cost are upper bounded,
\[
 \sum_{u \in V} b_u (f_{\hat{h}(u)} + d(u, \hat{h}(u))) - OPT \le 4 \delta n b_{avg}
\]
\label{lem:reconnect_cost}
\end{lemma}

\begin{proof}
Consider a location \(v \in I\) from the maximal independent set. Since \(v\) is in \(I\), it is opened as a facility in both the optimal solution and in the solution produced by Algorithm \ref{alg:reconn}. We begin by bounding the cost of reconnecting every node in \(B(v, \delta)\) to \(v\).

Take any \(u \in B(v,\delta)\) so \(d(u,v) \le \delta\). Let \(w \in M\) be such that \(h(u) = w\); that is, in the optimal solution \(u\) is connected to \(w\), but in the modified solution \(\hat{h}(u) = v\) (i.e. \(u\) is reconnected from $w$ to \(v\)).

In the optimal solution, the cost associated with \(u\) is 
$b_u\bigl(f_w + d(u,w)\bigr)$,
while in the reconnected solution the cost is 
$b_u\bigl(f_v + d(u,v)\bigr)$.
Because both \(v\) and \(w\) are marked in the optimal assignment, each prefer being connected to itself rather than to the other, which gives us:
\[
f_v < f_w + d(v,w) \quad \text{and} \quad f_w < f_v + d(v,w).
\]
Using these inequalities along with the triangle inequality, we can bound the reconnection cost:
\[
b_u\bigl(f_v+d(u,v)\bigr) < b_u\bigl(f_w+d(v,w)+d(u,v)\bigr) \le OPT_u + 2b_u\delta,
\]
when $OPT_u$ is the cost for $u$ in the optimal solution, i.e. $OPT_u = b_u(f_{h(u)} + d(u,h(u)))$. This shows that reconnecting the node $u$ from $w$ to $v$ increases the costs by at most $2b_u\delta$.

We now bound the cost of reconnecting nodes that are not within a \(\delta\)-distance of any location in \(I\). In Algorithm \(\ref{alg:reconn}\), the optimal assignment for such nodes is made to a node in \(I\). Let \(u \in V\), \(v \in I\), and \(w \in M\) be such that \(h(u) = w\) and \(\hat{h}(u) = v\). Moreover, assume that there is no \(x \in I\) satisfying \(d(x,u) \le \delta\); otherwise, we would have applied the previous case.

As before, we want to bound the cost incurred by \(u\) under the solution \((\hat{h}, \hat{k})\), which is \(b_u\,(f_v + d(v,u))\). Because \(u\) was reconnected to \(v\), it follows that \(w\) was not selected in the maximal independent set. By the properties of such a set, there must exist a neighbor of \(w\) in \(G\) that belongs to \(I\); denote this neighbor by \(v\). Hence, we have
$d(v,w) \le 2\delta$.

Using this, we obtain
\[
b_u\,(f_v + d(v,u)) \le b_u\,(f_w + d(w,u) + 2d(v,w)) \le OPT_u + 4b_u\delta.
\]
This shows that the extra cost for reconnecting node \(u\) is at most \(4\delta b_u\). Summing over all nodes in \(V\) gives an overall additional reconnection cost bounded by
$\sum_{u \in V} 4\delta b_u \le 4\delta\, n\, b_{avg}$.

\end{proof}

In Theorem \ref{theorem:no_reconn_costs} we use the fact that $\sum_{v \in M'} \sqrt{|L'_v|}f_v \le OPT$ for the analysis of the total expected costs of the set of facilities $M'$ opened by Algorithm \ref{alg:no_reconn}. For the set of opened facilities $I$ by Algorithm \ref{alg:reconn} this statement no longer holds. In the following we show that for $I$ an additional additive factor of $\delta n$ has to be introduced.

\begin{theorem} \label{theorem:recon_sum_bound}
    For the set of facilities $I$ opened by the reconnection algorithm, the sum of facility costs with exactly one client present at every location stays bounded:
    \[
    \sum_{v \in I} |\hat{L}_v| f_v \le OPT + \delta n
    \]    
\end{theorem}
\begin{proof}
    We show the statement by proofing that a location introduces costs of at most its part in OPT plus $\delta$ in the left hand side.
    Let $u \in V$, $v \in I$, $w \in M$ such that $h(u) = w$ and $\hat{h}(u) = v$. The costs of $u$ in the left hand side are $f_v$ while on the right hand side it is at least $f_w + d(w, u)$. If $u$ was not reconnected, so $v = w$, the claim follows.
    For the case of a reconnection of $u$, we show, $f_v \le f_w + d(u, w) + \delta$. There are two cases in which $u$ can be reconnected to $v$. Firstly, $u$ is close to the location $v$: $u \in B(v, \delta)$.
    \[
        f_v \le f_w + d(w, v) \le f_w + d(w, u) + d(u, v) \le f_w + d(w, u) + \delta
    \]
    The first inequality, follows from the optimality before reconnection and the last one from $u \in B(v, \delta)$.
    Secondly, if $u \notin B(v, \delta)$ then $u$ was reconnected to $v$ because it yielded the lowest costs:
    \[
        v = \argmin_{x \in I} f_x + d(x, u) 
    \]
    The maximal independent set algorithm adds locations in ascending order of facility costs. This means if $w \notin I$ one of its neighbors $x \in I$ with $f_x \le f_w$ was chosen instead.
    \[
        f_v \le f_x + d(u,x) - d(u,v) \le f_x + d(u,w) + d(w, x) - d(u,v) \le f_w + d(u,w) + \delta
    \]
    The first equality follows from the optimality after reconnection and the last one follows from $u \notin B(v, \delta)$ and $d(w,x) \le 2\delta$.
    Therefore, in all cases the costs of $u$ on the left hand side are bounded by its part in the optimal solution $OPT_u$ plus $\delta$. For the sum of $n$ locations this yields a bound of $OPT + \delta n$.
\end{proof}

We are now ready to demonstrate the main statement of this paper.

\begin{theorem} \label{theorem:error_main}
    Algorithm \(\ref{alg:reconn}\) has a failure probability of at most \(\alpha\). Moreover, when it succeeds, its expected cost is bounded by  
\[
\left(1 + \frac{2}{\varepsilon}\ln\frac{2n}{\alpha}\frac{1}{\gamma \ln n}\right)OPT + \delta n (4b_{avg} + \frac{2}{\varepsilon}\ln{\frac{2n}{\alpha}\frac{1}{\gamma\ln{n}}}).
\] \label{thm:main}
\end{theorem}

\begin{proof}
The failure probability is established using an argument analogous to that in Theorem \ref{theorem:no_reconn_failure}. Moreover, by Lemma \(\ref{lem:sizeLv}\) and Theorem \ref{theorem:recon_sum_bound} we have
\[
\sum_{v\in I}\sqrt{|\hat{L}_v|}f_v \le \frac{1}{\gamma \ln n}\sum_{v\in I}|\hat{L}_v|f_v \le \frac{1}{\gamma \ln n}\, (OPT + \delta n).
\]
Following a similar reasoning as in Theorem \(\ref{theorem:no_reconn_costs}\), the lemma statement then follows.
\end{proof}
When \(\varepsilon\), \(\alpha\), and \(\gamma\) are constants, the multiplicative factor  
$\left(1 + \frac{2}{\varepsilon}\ln\frac{2n}{\alpha}\frac{1}{\gamma \ln n}\right)$
remains \(O(1)\). Additionally, if the locations are sufficiently dense (i.e., \(\delta\) is small), the additive term \(\delta n (4b_{avg} + \frac{2}{\varepsilon}\ln{\frac{2n}{\alpha}\frac{1}{\gamma\ln{n}}})\) becomes negligible compared to \(OPT\).
While we assume \( b_v \geq 1 \) for all \( v \in V \) in this section, we present an analysis for the case where \( b_v \geq 0 \) in the following one.

\section{Additional Theoretical Results for $b_v \geq 0$} \label{chapter:other_results}
In this section we provide an analysis of our $\varepsilon$-LDP algorithm \ref{alg:reconn} that does not require the presence of at least one client at every location.
Until now, we assumed that every location hosts at least one client (\(b_v \ge 1\) for all \(v\in V\)). This assumption was essential for establishing \(\sum_{v\in M'}\sqrt{|L'_v|}f_v \leq OPT\) in Section 4 and \(\sum_{v\in I}\sqrt{|\hat{L}_v|}f_v \leq (OPT+ \delta n)/(\gamma \ln n)\) in Section 5. Without this assumption, scenarios with many facilities but only one client per location would lead to a poor approximation, since the optimum ($OPT$) depends on the number of clients, while the costs of the private algorithm depend on the number of locations.

In this section, we relax the assumption that every location has at least one client. In particular, we now allow \(b_v \in \mathbb{N}_{\ge 0}\). 
We introduce two different approaches. The first one assumes a total ratio $\nu$ between clients and locations, while the second approach assumes that a location has client presence with a probability of at least $p$.

\subsection{Client-Location Ratio Analysis}
In the following we define \(\nu = \frac{N}{n}\) as the ratio between the total number of clients $N$ and the number of locations $n$.
Moreover, we assume that the ratio \(\eta = \frac{f_{max}}{f_{min}}\) is constant.
As discussed previously, we need a revised bound on \(\sum_{v \in I}\sqrt{|\hat{L}_v|}f_v\). We observe that
\[
\sum_{v \in I}\sqrt{|\hat{L}_v|}f_v \le \frac{f_{max}}{\gamma \ln n}\sum_{v \in I}|\hat{L}_v| = \frac{f_{max}}{\gamma \ln n}\frac{N}{\nu} \le \frac{1}{\gamma \ln n} \frac{\eta}{\nu}\, OPT.
\]
Using this inequality, we follow the arguments in the proof of Theorem \ref{thm:main} to derive an upper bound on the total expected cost:
$\left(1 + \frac{2}{\varepsilon}\ln\frac{2n}{\alpha}\frac{1}{\gamma\ln{n}}\frac{\eta}{\nu}\right)OPT + 4\delta n b_{avg}$.
When $\varepsilon$, $\alpha$, $\gamma$, $\eta$, and $\nu$ are constants, we obtain that the multiplicative factor $\left(1 + \frac{2}{\varepsilon}\ln\frac{2n}{\alpha}\frac{1}{\gamma\ln{n}}\frac{\eta}{\nu}\right)$ remains $O(1)$.

\subsection{Bernoulli Distributed Presence Analysis}
In this section we assume the probability of the presence of at least one client at a location is lower bounded. 
\paragraph{Client distribution assumption}
Let $v \in V$, we assume that with constant probability $p$ at least one client is present at $v$:
\[Pr[b_v \ge 1] \ge p\]

In this scenario we propose two different analysis for our reconnection algorithm. The first one assumes $|\hat{L}_v| \ge \gamma^2\ln^2{n}$ and $p > \frac{1}{\gamma}$ while the second one requires a stronger assumption about the distribution of locations with $|\hat{L}_v| \ge \gamma^2\ln^3{n}$ but drops the $p > \frac{1}{\gamma}$ requirement.
For the analysis we remind the reader of Hoeffding's inequality.
\begin{theorem}[Hoeffding's Inequality \cite{hoeffding1963probability}]\label{theorem:hoeff}
    Suppose $X_1, ...,  X_n$ are independent random variables taking values in $[a, b]$. Let $X = \sum_{i=1}^{n}X_i$ denote the sum and let $\mu = E[X]$ denote the expected value of the sum. Then, for $t > 0$,
    \begin{equation}
        Pr[X \le \mu - t] < exp(-2t^2/n(b-a)^2)
    \end{equation}
\end{theorem}

We now establish the relationship between the number of connected clients $N_v = \sum_{u \in \hat{L}_v}b_u$ and number of connected locations $|\hat{L}_v|$ for a facility $v \in I$ for the setting of $|\hat{L}_v| \ge \gamma^2 \ln^2{n}$. 
\begin{theorem}\label{theorem:client_pres}
    For a facility $v \in I$ opened by Algorithm \ref{alg:reconn}, $p > 1/\gamma$, and $n$ sufficiently large:
    \[
        Pr[\sqrt{L_v}\le \frac{1}{\ln{n}} N_v] \ge 1- 1/n^2
    \]
\end{theorem}
\begin{proof}
    Let $v \in I$ be a facility opened by the reconnection algorithm. We introduce the indicator variable $X_u = 1\{b_u \ge 1\}$ for a connected location $u \in \hat{L}_v$ to indicate the presence of at least one client.
    \[
        Pr[\sqrt{|\hat{L}_v|} \le \frac{1}{\ln{n}}N_v] \ge Pr[\sqrt{|\hat{L}_v|} \le \frac{1}{\ln{n}}\sum_{u \in \hat{L}_v}X_u]
    \]
    With our assumption about the client distribution we get $X_u \sim Ber(p)$ for all $u \in \hat{L}_v$ and $E[\sum_{u \in \hat{L}_v} X_u] = |\hat{L}_v|p$.
    For our bound we use Hoeffding's inequality from Theorem \ref{theorem:hoeff} to bound the failure probability $Pr[\sum_{u \in \hat{L}_v}X_u \le \ln{n}\sqrt{|\hat{L}_v|}]$.
    We rewrite this probability to fit the framework of Hoeffding's inequality. We need to find $t > 0$ such that $p|\hat{L}_v| - t = \ln{n}\sqrt{|\hat{L}_v|}$. This yields
    \[
        t = p|\hat{L}_v| - \ln{n} \sqrt{|\hat{L}_v|}
    \]
    Since $p > \frac{1}{\gamma}$ is constant and $|\hat{L}_v| \ge \gamma^2\ln^2{n}$, $t > 0$ is satisfied.
    Hence, we can apply Hoeffding's inequality with $X_u \in [0,1]$ for all $u \in \hat{L}_v$.
    \[
        Pr[\sum_{u \in \hat{L}_v}X_u \le p|\hat{L}_v| - t] \le exp(-2|\hat{L}_v|(p - \frac{\ln{n}}{\sqrt{|\hat{L}_v|}})^2) 
    \]
    With $p > \frac{1}{\gamma}$, $(p - \frac{\ln{n}}{\sqrt{|\hat{L}_v|}})^2$ is lower bounded by the constant $(p-\gamma)^2$ for all $n$.
    \[
        exp(-2|\hat{L}_v|(p - \frac{\ln{n}}{\sqrt{|\hat{L}_v|}})^2) \le exp(-2|\hat{L}_v|(p-\gamma)^2) \le \frac{1}{n^2}
    \]
    The second inequality holds for sufficiently large $n$ and therefore concludes this proof.
\end{proof}
We now apply Theorem \ref{theorem:client_pres} to establish a bound for all facilities $v \in I$.
\begin{theorem}\label{theorem:total_cl_dis}
    With probability $1-\frac{1}{n}$, $p > \frac{1}{\gamma}$ and for sufficiently large n, for all $v \in I$
    \[
        \sqrt{|\hat{L}_v|}\le \frac{1}{\ln{n}} N_v
    \]
\end{theorem}
\begin{proof}
    Follows from Theorem \ref{theorem:client_pres} and the union bound for at most $n$ facilities opened by Algorithm \ref{alg:reconn}.
\end{proof}

With the relationship between number of connected locations and number of connected client we just established, we can bound the expected costs of the Algorithm \ref{alg:reconn} in the setting of Bernoulli distributed client presence.

\begin{theorem} \label{theorem:err_bern}
    Algorithm \ref{alg:reconn} has a failure probability of at most $\alpha$. When it succeeds, with probability of at least $1-\frac{1}{n}$ for large enough $n$ the expected costs are bounded by $ (1 + \frac{2}{\varepsilon}\ln{\frac{2n}{\alpha}}\frac{1}{\ln{n}})OPT + 2n\delta(\frac{2}{\varepsilon}\ln{\frac{2n}{\alpha}} +  b_{avg})$.
\end{theorem}
\begin{proof}
    The failure probability follows from the same argument as in Theorem \ref{theorem:no_reconn_failure} since it does not make assumptions about the client distribution.
    The reconnection costs are still bounded by $4 \delta nb_{avg}$ as shown in Lemma \ref{lem:reconnect_cost}.
    We now bound the costs of the additional margin: $\frac{2}{\varepsilon}\ln{\frac{2n}{\alpha}\sum_{v \in I}\sqrt{|\hat{L}_v|}f_v}$.
    From Theorem \ref{theorem:total_cl_dis} it follows that with probability $1-\frac{1}{n}$,
    \[
     \sum_{v \in I}\sqrt{|\hat{L}_v|}f_v \le \sum_{v \in I}\frac{1}{\ln{n}}N_vf_v \le \frac{1}{\ln{n}}(OPT + n\delta b_{avg})
    \]
    This results in the bound of total expected costs of 
    \[
        (1 + \frac{2}{\varepsilon}\ln{\frac{2n}{\alpha}}\frac{1}{\ln{n}})OPT + n\delta b_{avg}(\frac{2}{\varepsilon}\ln{\frac{2n}{\alpha}}\frac{1}{\ln{n}} + 4)
    \]
\end{proof}

In the following we provide an analysis such that Algorithm \ref{alg:reconn} achieves the same multiplicative error bound as in Theorem \ref{theorem:error_main} of $(1 + \frac{2}{\epsilon}\ln{\frac{2n}{\alpha}\frac{1}{\gamma \ln{n}}})$ in the setting of Bernoulli distributed clients under a stronger version of Equation \ref{eq:local}. 
\begin{theorem}\label{theorem:fail_bern}
    For a facility $v \in I$ opened by Algorithm \ref{alg:reconn}, $|\hat{L}_v| \ge \gamma^2\ln^3{n}$, $p$ constant and for $n$ large enough.
    \[
        Pr[\sqrt{|\hat{L}_v|} \le \frac{1}{\gamma\ln{n}}N_v] \ge 1-\frac{1}{n^2}
    \]
\end{theorem}
\begin{proof}
    The proof follows from the same arguments used in Theorem \ref{theorem:client_pres}. We choose $t = p|\hat{L}_v| - \gamma \ln{n} \sqrt{|\hat{L}_v|}$. From $|\hat{L}_v| \ge \gamma^2\ln^3{n}$ and $p$ constant, $t > 0$ is satisfied. We apply Hoeffding's inequality to get
    \[
        Pr[\sum_{u \in \hat{L}_v} X_u \le p|\hat{L}_v| - t] \le exp(-2|\hat{L}_v|(p - \frac{\gamma \ln{n}}{\sqrt{|\hat{L}_v|}})^2)
    \]
    For large enough n, this probability is upper bounded by $1/n^2$.
\end{proof}
We can now proof the constant approximation ratio in this setting.
\begin{theorem}
    Under the assumption of $|B(v, \delta)| \ge \gamma^2 \ln^3{n}$ for every $v \in V$, Algorithm \ref{alg:reconn} has a failure rate of $\alpha$. If the algorithms succeeds, with probability of at least $1 - \frac{1}{n}$ for large enough $n$ the expected costs are bounded by $(1 + \frac{2}{\varepsilon}\ln{\frac{2n}{\alpha}\ln{\frac{1}{\gamma \ln n}}})OPT + n\delta b_{avg}(\frac{2}{\varepsilon}\ln{\frac{2n}{\alpha}}\frac{1}{\gamma \ln{n}} + 4)$
\end{theorem}

\begin{proof}
    The proof follows from Theorem \ref{theorem:total_cl_dis} and \ref{theorem:err_bern} in combination with the new bound proved in Theorem \ref{theorem:fail_bern}.
\end{proof}

\section{Experimental Results}
In this section, we evaluate the private algorithms with various parameter settings on both synthetically generated and real-world datasets, comparing their performance against the non-private algorithm. We generate synthetic instances using two distinct methods. The first employs the Matérn cluster point process~\cite{matern,poisson_point}, which creates clustered instances where each cluster simulates a densely populated neighborhood. By adjusting the generation parameters, we can control both the number of neighborhoods (centers) and the number of households (locations) within each neighborhood. The second method uses a Poisson point process, where the number of locations is drawn from a Poisson distribution, and these locations are then uniformly distributed across a simulation window. We also provide experiments for a dataset based on the city of Chiang Mai, Thailand.
Our results show that for all instances, there exists a value of \(\delta\) such that the private reconnection algorithm outperforms the straightforward approach.

\subsection{Synthetic Instances}

We generate the locations’ positions using the Matérn cluster point process \cite{matern}.
The process takes the tuple \((n, \gamma, \delta_{\text{gen}})\) as input, where \(n\) is the expected total number of locations, \(\gamma\) is a scaling parameter, and \(\delta_{\text{gen}}\) defines the clustering radius. Let \(n_{\text{centers}}\) be the number of centers generated and \(n_{\text{daughter}}^i\) the number of locations around center \(i\). We require two conditions:\newline
1. Each center should have at least \(\gamma^2 \ln^2 n\) locations within \(\delta_{\text{gen}}\) in expectation, i.e.
   $\mathbb{E}[n_{\text{daughter}}^i] \ge \gamma^2 \ln^2 n$. \newline
2. The total expected number of locations should be \(n\), i.e.
   $\mathbb{E}\left[\sum_{i=1}^{n_{\text{centers}}} n_{\text{daughter}}^i\right] = n$.

We model \(n_{\text{daughter}}^i\) as a Poisson random variable with parameter \(\lambda_{\text{daughter}}\) and \(n_{\text{centers}}\) as a Poisson random variable with parameter \(\lambda_{\text{centers}}\). Since the expected value of a Poisson distribution equals its \(\lambda\)-parameter, we set
$\lambda_{\text{daughter}} = \gamma^2 \ln^2 n$,
so that \(\mathbb{E}[n_{\text{daughter}}^i] = \gamma^2 \ln^2 n\). To ensure that the total expected number of locations is \(n\), we choose
$\lambda_{\text{centers}} = \frac{n}{\lambda_{\text{daughter}}}$,
since then
$\mathbb{E}[n_{\text{centers}}] \cdot \mathbb{E}[n_{\text{daughter}}^i] = \lambda_{\text{centers}} \cdot \lambda_{\text{daughter}} = n$.

The process first samples \(n_{\text{centers}} \sim \text{Poisson}\left(\frac{n}{\gamma^2 \ln^2 n}\right)\) and distributes these centers uniformly at random on a \(1 \times 1\) simulation window. For each center, it samples \(n_{\text{daughter}} \sim \text{Poisson}(\gamma^2 \ln^2 n)\). Then, for each location, a radial coordinate is drawn uniformly from \([0, \delta_{\text{gen}}]\) and an angular coordinate from \([0, 2\pi]\), which are subsequently converted to Cartesian coordinates. Because each location lies at most \(\delta_{\text{gen}}\) away from its center, the overall simulation window expands to \((1+2\delta_{\text{gen}}) \times (1+2\delta_{\text{gen}})\).

We generate the number of clients per location from a Gaussian distribution with a mean of \(2.5\) and a standard deviation of \(1.5\). The resulting \(b_v\) values are then rounded to the nearest integer and restricted to the interval \([0, 8]\).
For the facility costs at each location, we draw values from a uniform distribution over a specified interval. 

Figure \ref{fig:cluster_ins} depicts example instances generated by the Matérn cluster point process with varying $\gamma$ values. The number of centers $n_{centers}$ depends inversely on $\gamma$. Adjusting $\gamma$ allows for the simulation of different location distributions.

\begin{figure}[h]
    \centering
    \begin{subfigure}{0.45\textwidth}
        \centering
        \includegraphics[width=\linewidth]{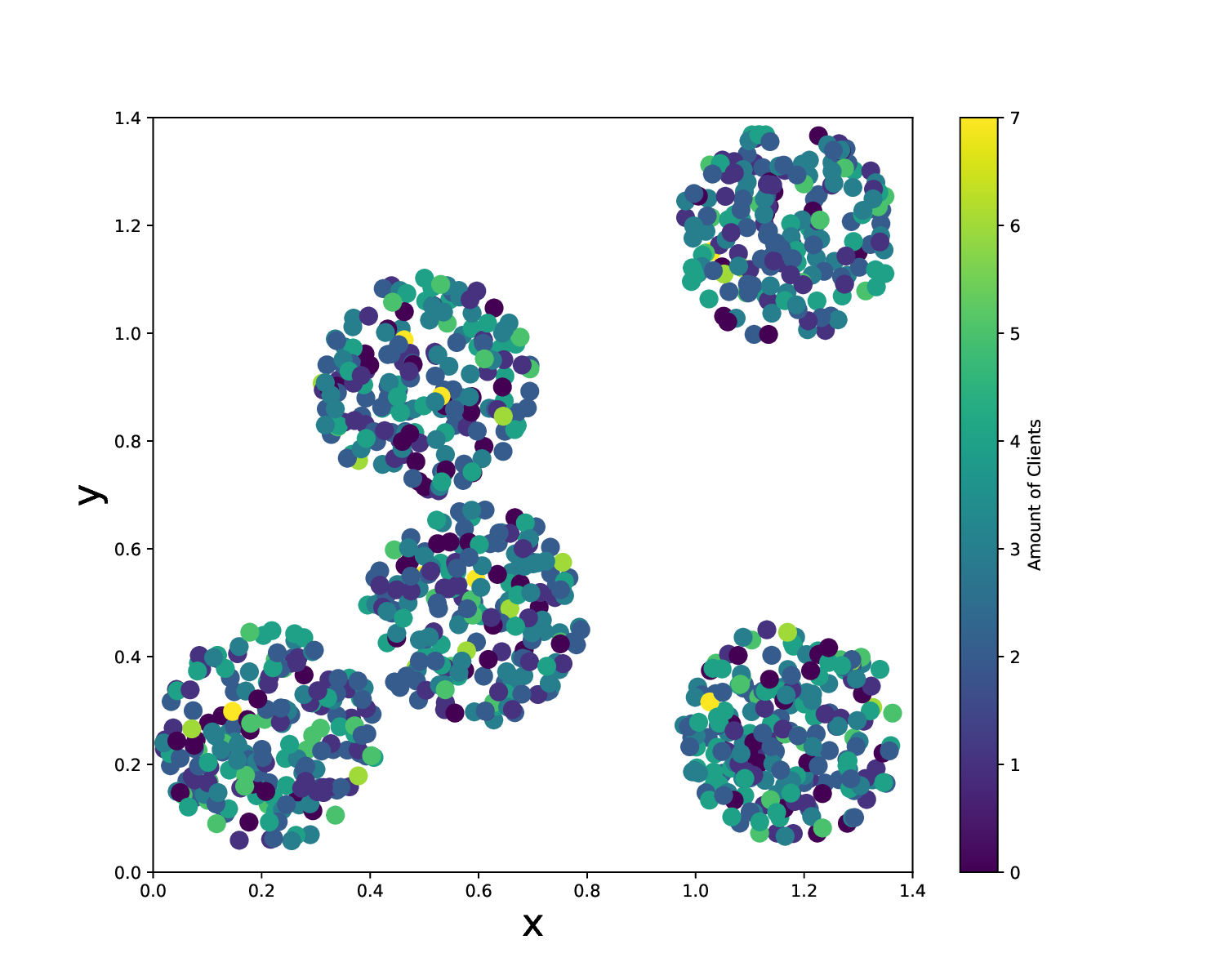}
        \caption{$\gamma = 2$}
    \end{subfigure}
    \begin{subfigure}{0.45\textwidth}
        \centering
        \includegraphics[width=\linewidth]{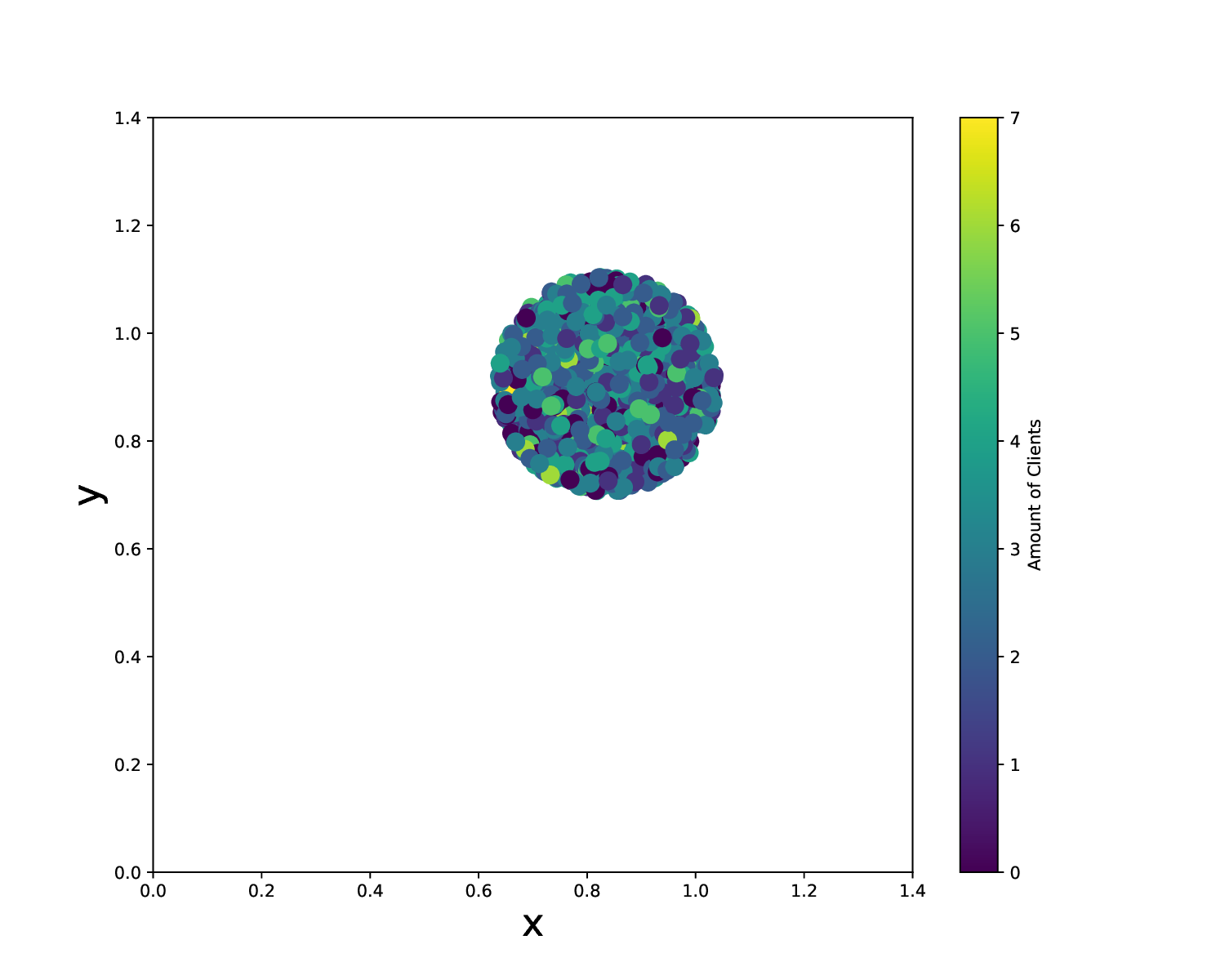}
        \caption{$\gamma = 5$}
    \end{subfigure}
    \caption{Clustered instances created by Matérn cluster point process}
    \label{fig:cluster_ins}
\end{figure}

\paragraph{Our Results}
Figure \ref{fig:delta_plots} shows the normalized costs of the private reconnection algorithm in comparison with the optimal non-private and straightforward private algorithm. In this benchmark $\delta$ is increased from 0 to 1 with a step size of $0.01$ for the private reconnection algorithm. 
For every $\delta$, 1000 instances are generated with $n = 1000, \gamma = 2, \delta_{gen}=0.2$. The private algorithms are executed with $\varepsilon = 0.1$ and $\alpha = 0.1$. It can be seen that for clustered instances the reconnection algorithm outperforms the straightforward approach for any $\delta$.
Furthermore, the reconnection algorithm performs better in comparison to the straightforward approach if no locations with facility costs close to 0 exist. For locations with facility costs of almost 0 it is more likely that they connect all of the other close locations anyways. Therefore, the reconnection part reconnects fewer locations leading to more similar solutions.

\begin{figure} [h]
    \centering
    \begin{subfigure}{0.40\textwidth}
        \centering
        \includegraphics[width=\linewidth]{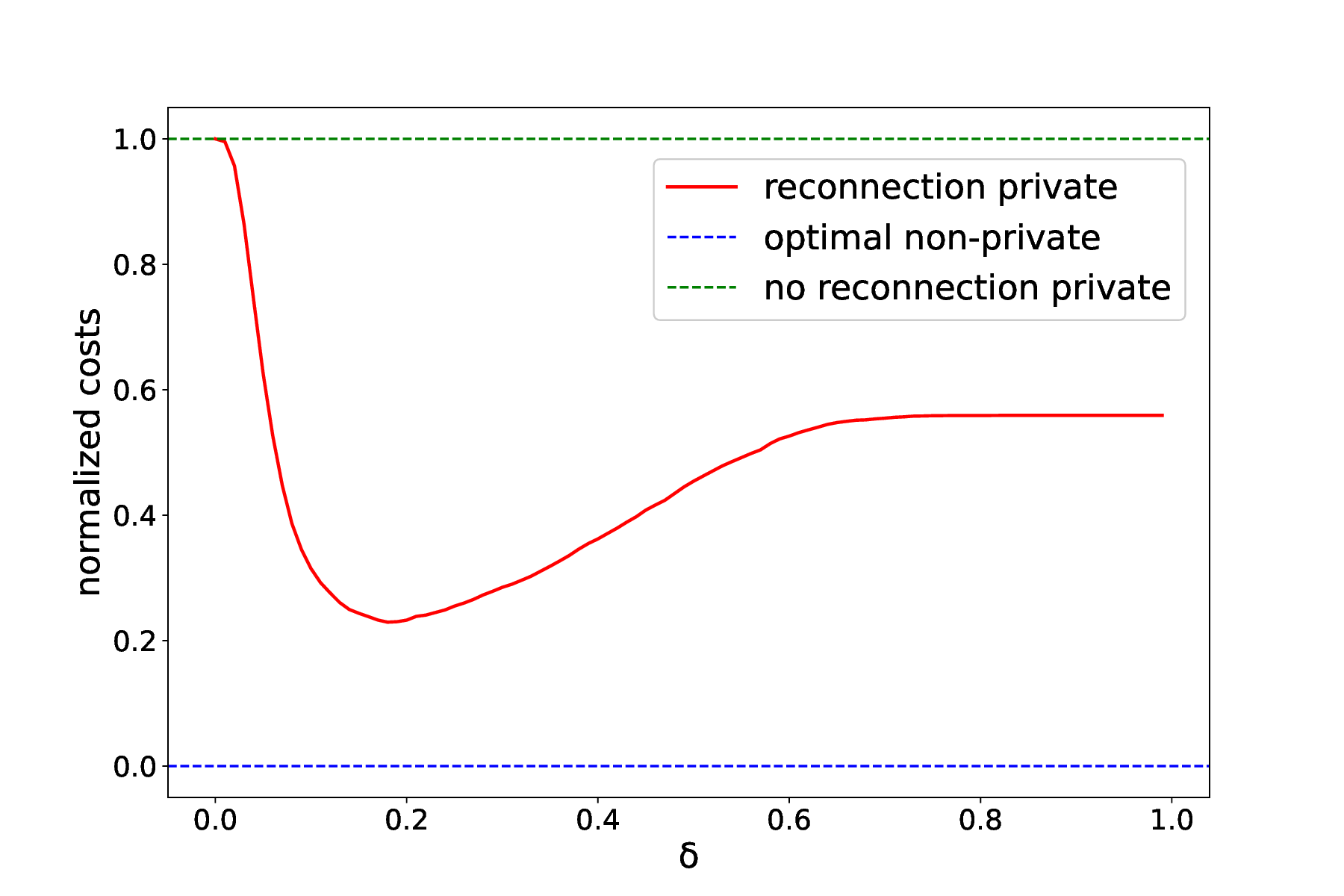}
        \caption{$f_v \in [0,1]$}
    \end{subfigure}
    \begin{subfigure}{0.40\textwidth}
        \centering
        \includegraphics[width=\linewidth]{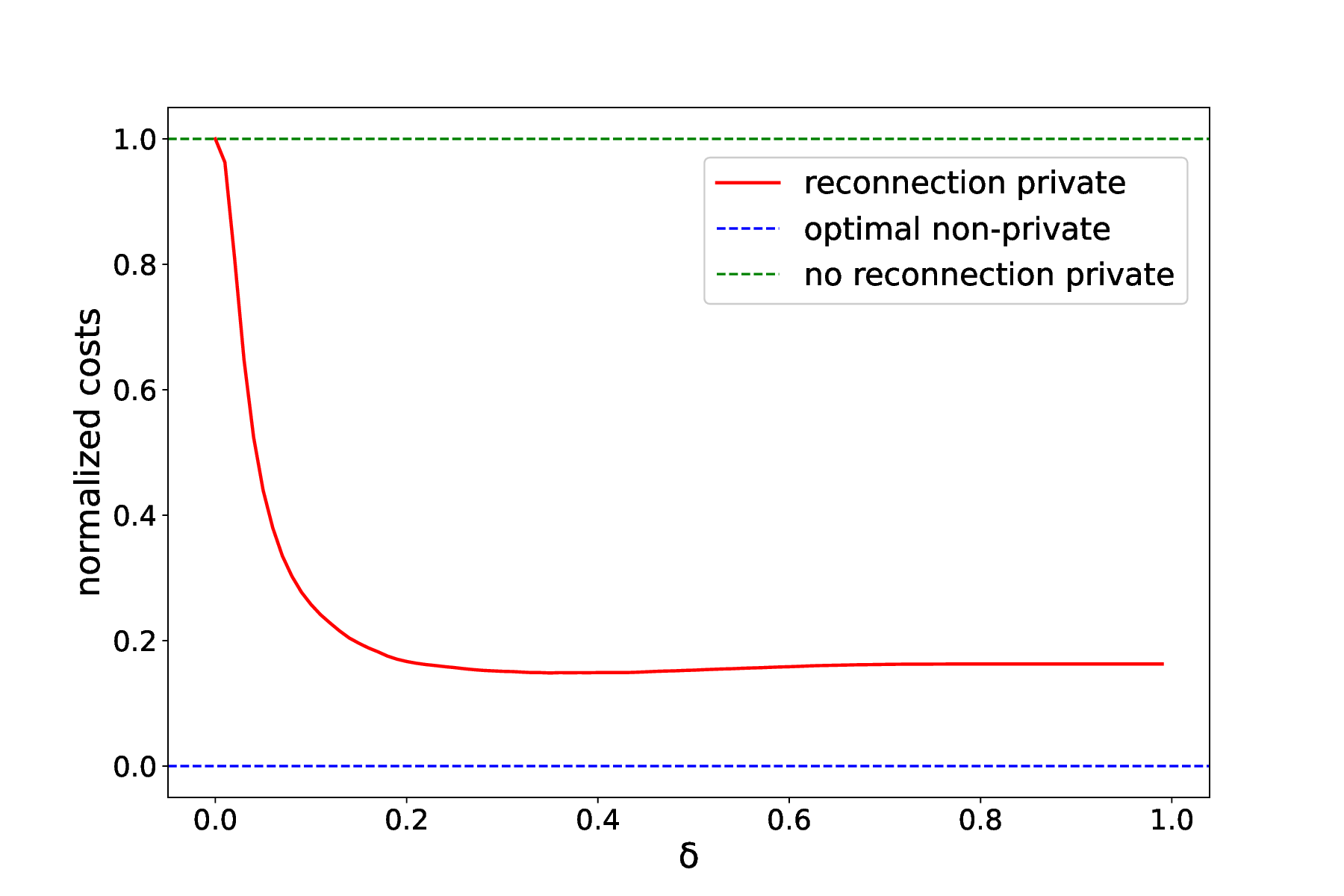}
        \caption{$f_v \in [0.1, 0.3]$}
    \end{subfigure}
    \begin{subfigure}{0.40\textwidth}
        \centering
        \includegraphics[width=\linewidth]{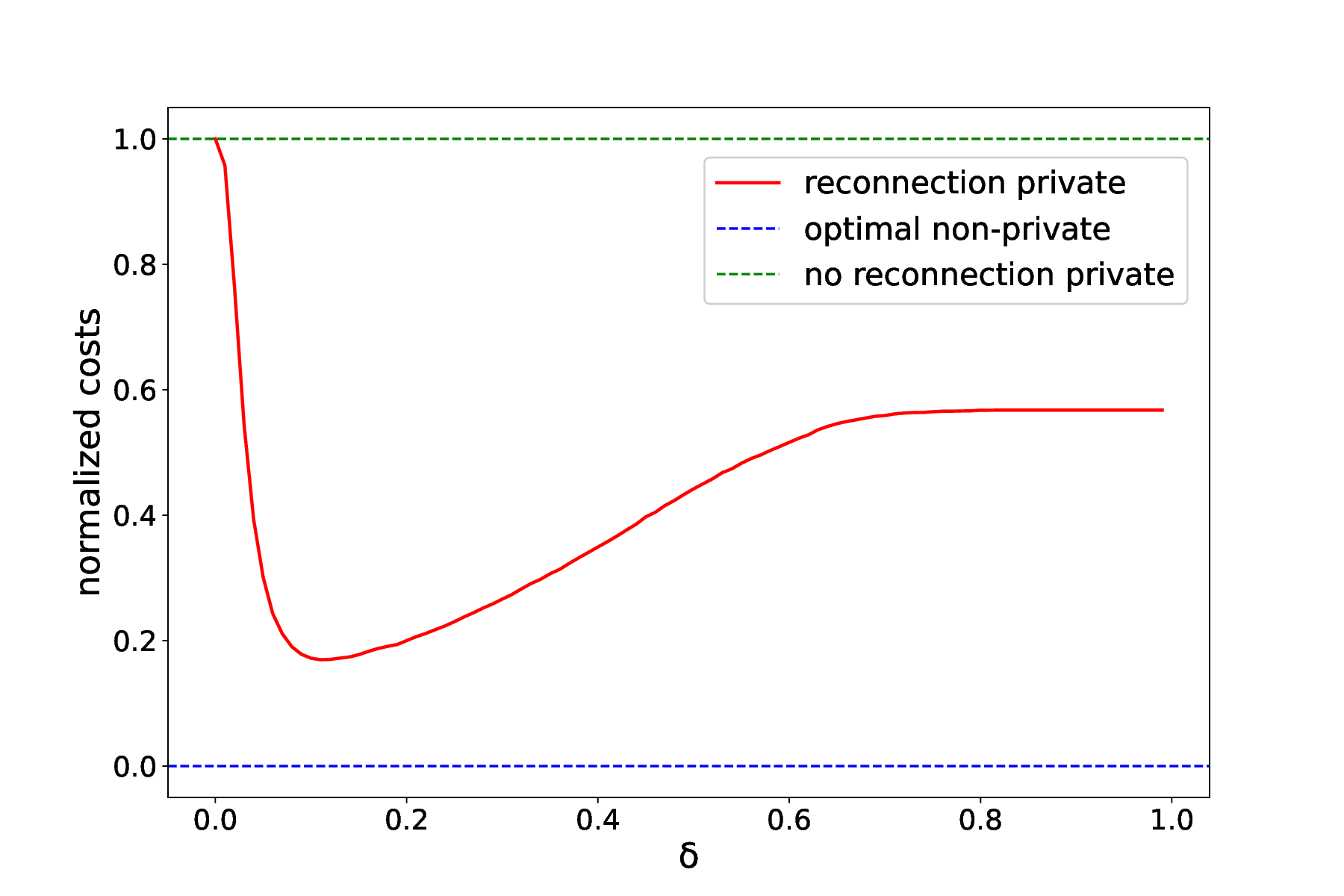}
        \caption{$f_v \in [0, 0.2]$}
    \end{subfigure}
    \caption{Normalized costs for varying $\delta$ and facility cost ranges on clustered instances}
    \label{fig:delta_plots}
\end{figure}
\begin{figure}[h]
    \centering
    \begin{subfigure}[b]{0.43\linewidth}
        \centering
        \includegraphics[width=\linewidth]{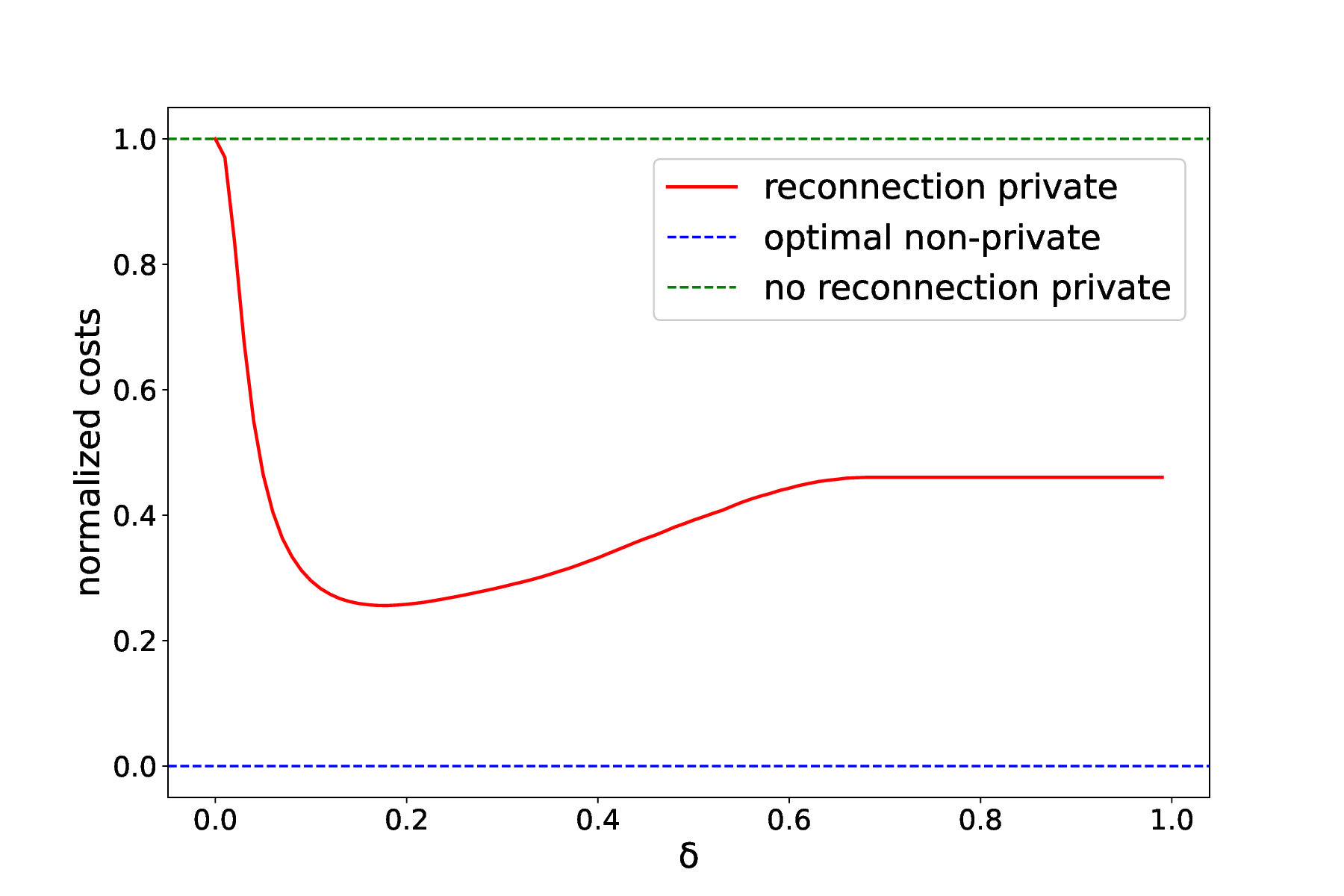}
        \caption{}
        \label{fig:uniform}
    \end{subfigure}
    \hfill
    \begin{subfigure}[b]{0.41\linewidth}
        \centering
        \includegraphics[width=\linewidth]{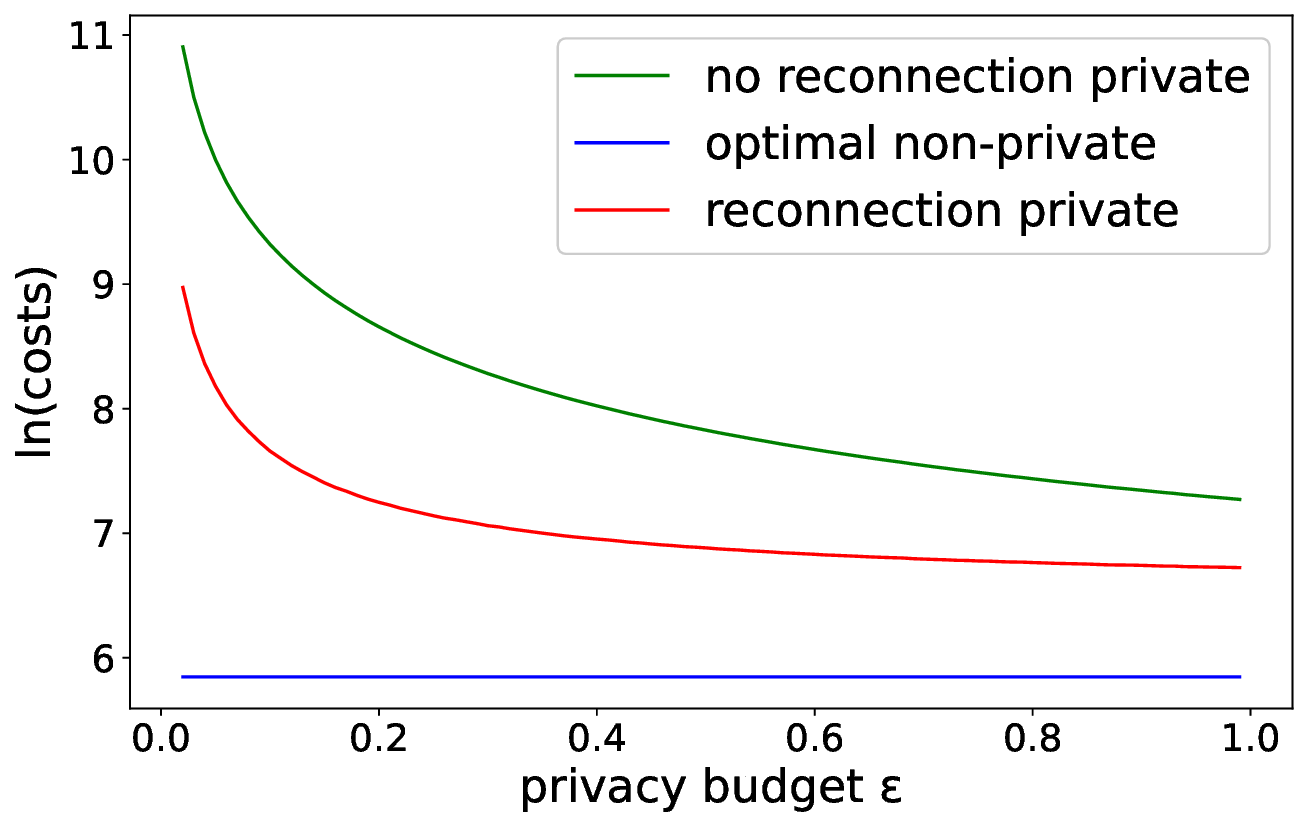}
        \caption{}
        \label{fig:eps_benchmark}
    \end{subfigure}
    \caption{(a) Normalized costs against varying \(\delta\) for Poisson point process instances (b) Impact of different privacy budgets \(\varepsilon\) on clustered instances}
    \label{fig:performance_evaluation}
\end{figure}

Figure \ref{fig:uniform} depicts the performance of the private reconnection algorithm on instances generated by a Poisson point process \cite{poisson_point}. 
The generation process first samples the number of locations $n_{locations}$ according to a Poisson distribution with parameter $\lambda = n$ and then generates $n_{locations}$ locations uniformly at random distributed on a $1\times1$ simulation window.

In Figure \ref{fig:eps_benchmark} the algorithms are compared for an varying privacy budget $\varepsilon$ from 0.01 to 1 with a step width of 0.001. For every $\varepsilon$, 100 instances are generated with $n = 1000, \gamma = 2, \delta_{gen}=0.2, f_v\in[0.1,0.3]$. The private algorithms are executed with $\alpha = 0.1$ and the reconnection algorithm uses $\delta = 0.2$.
The outcome aligns with the theoretical results. 
A smaller $\varepsilon$ leads to a larger coefficient $\frac{2}{\varepsilon}$ in the multiplicative approximation ratio. The reconnection algorithm keeps the influence of this coefficient small by shifting the impact of $\varepsilon$ into the small additive error. 

Figure \ref{fig:n_benchmark} shows the costs depending on the size of the instances. For $n \in [100, 5000]$ with a step size of 100, 500 instances with $\delta_{gen} = 0.2, \gamma=2$ and $f_v\in [0.1,0.3]$ are generated. The private algorithms are executed with $\varepsilon = 0.1$ and $\alpha = 0.1$. Furthermore, the reconnection algorithm is also executed with $\delta = 0.2$. In alignment with the previous results, the reconnection algorithm outperforms the straightforward algorithm. 

\begin{figure}[htbp]
    \centering
    \includegraphics[width=0.5\linewidth]{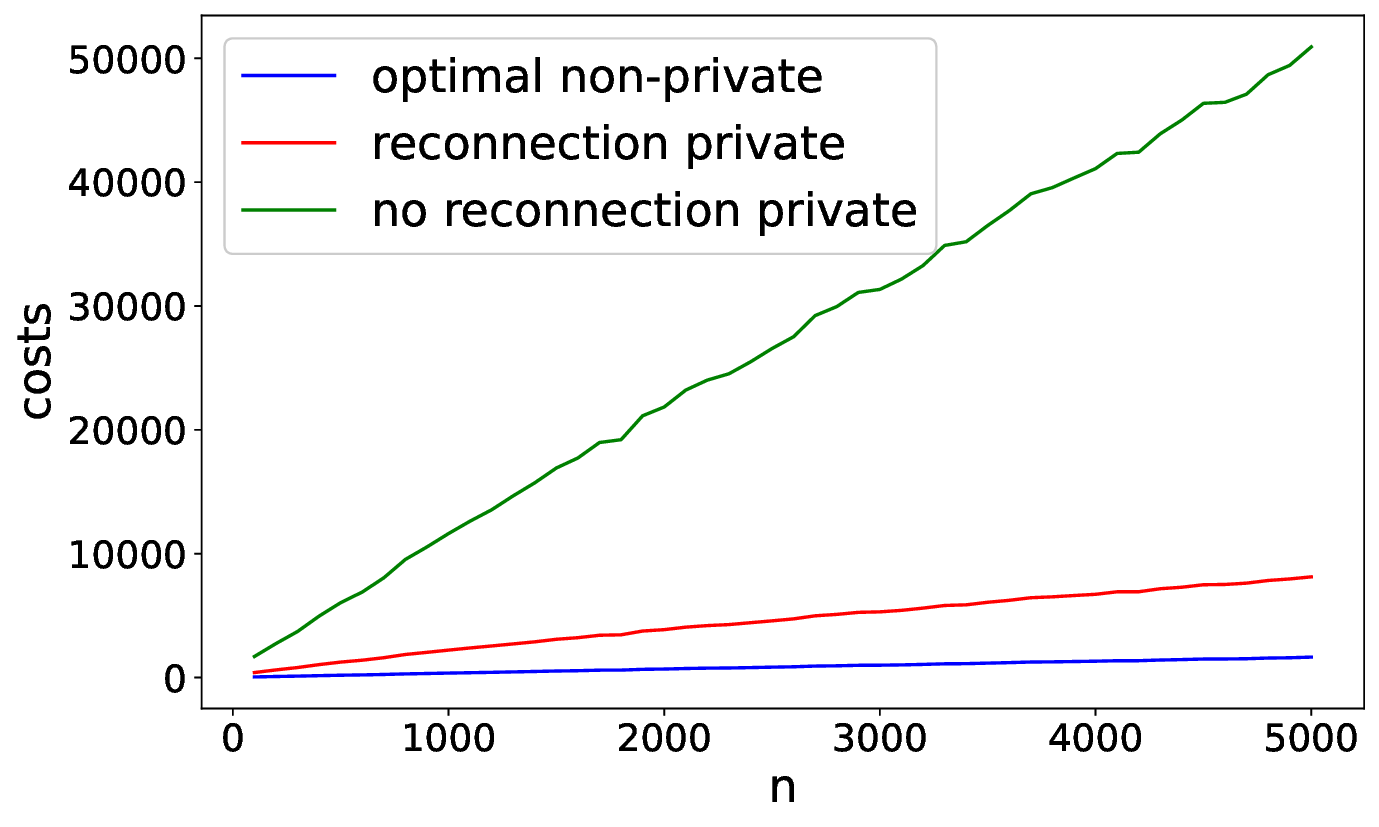}
    \caption{Costs for varying \(n\) on clustered instances}
    \label{fig:n_benchmark}
\end{figure}

\begin{figure}[h]
    \centering
    \includegraphics[width=0.5\linewidth]{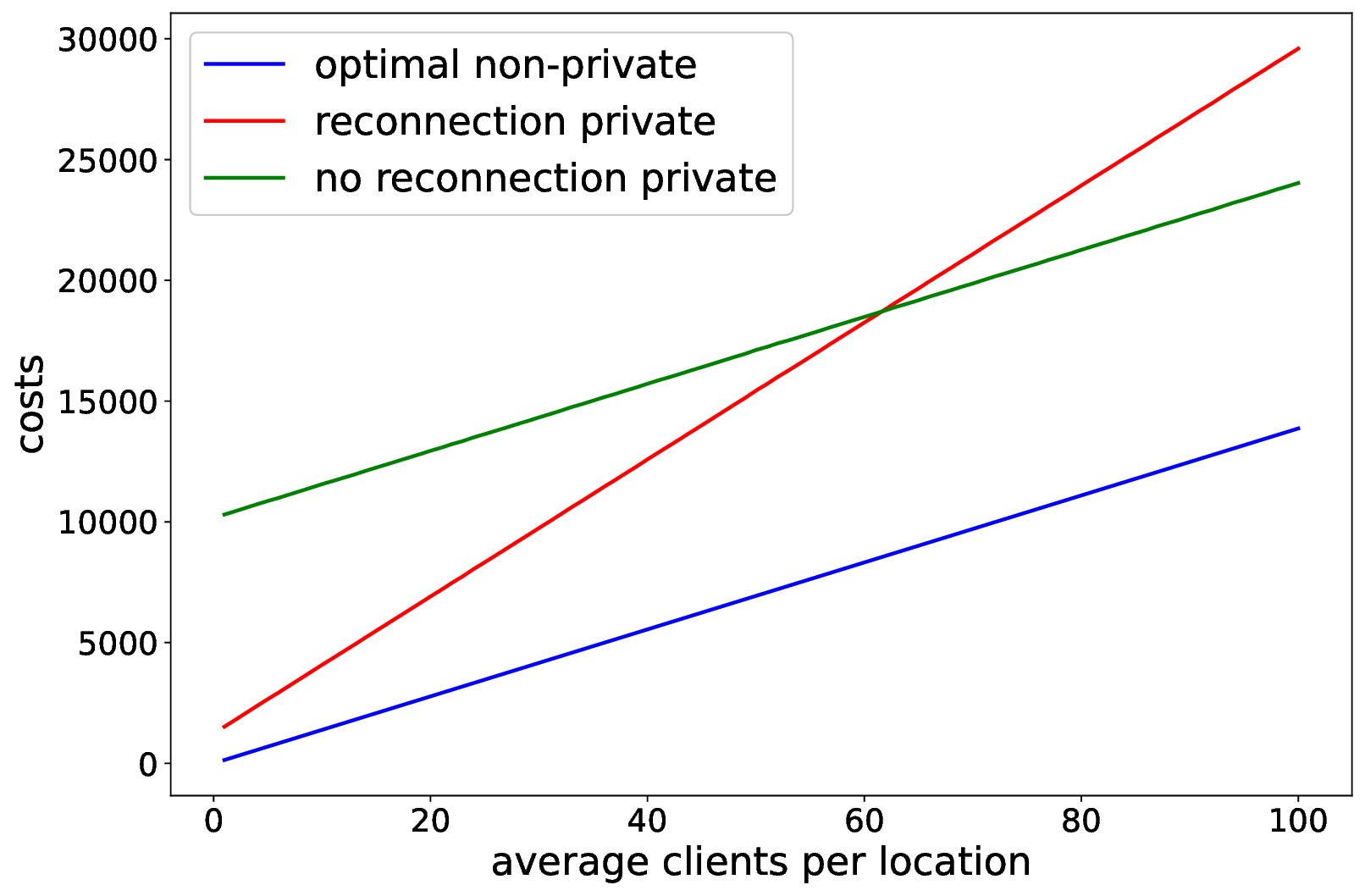}
    \caption{Costs for varying \(b_{avg}\) on clustered instances}
    \label{fig:b_benchmark}
\end{figure}

Figure \ref{fig:b_benchmark} presents the algorithm costs as the average number of clients per location, $b_{avg}$, increases. For this analysis, 100 instances are generated with parameters $n = 1000$, $\delta_{gen} = 0.2$, and $\gamma = 2$. 
In each instance, every location is assigned the same number of clients, $b_{avg}$. The algorithms are executed for values of $b_{avg}$ ranging from 1 to 100, and the costs are averaged across all instances. The private algorithms are run with parameters $\alpha = 0.1$, $\varepsilon = 0.1$, and $\delta = 0.2$.  

The results illustrate the impact of $b_{avg}$ on the additive costs of the private reconnection algorithm \ref{alg:reconn}. 
Additionally, they indicate that for scenarios where the average number of clients is below 75 (e.g., individuals within a household), our algorithm outperforms the straightforward approach.

\subsection{Real-World Instances}
For an application of the algorithm to the real world we generate a set of data according to the technique described in ``Submodularity Property for Facility Locations of Dynamic Flow Networks'' \cite{real_world_data}. It uses data from the project ``Urban Observatory and Citizen Engagement by Data-driven and Deliberative Design: A Case Study of Chiang Mai City'' to generate a set of locations with clients. We expand this data by first normalizing the position of the locations to a $1 \times 1$ window and then uniformly at random assign facility costs from the interval $[0.1, 0.3]$. The dataset consists of 431 different locations. Figure \ref{fig:real_world_instance} depicts the Chiang Mai location and client distribution.
\begin{figure}[htbp]
    \centering
    \includegraphics[width=0.5\linewidth]{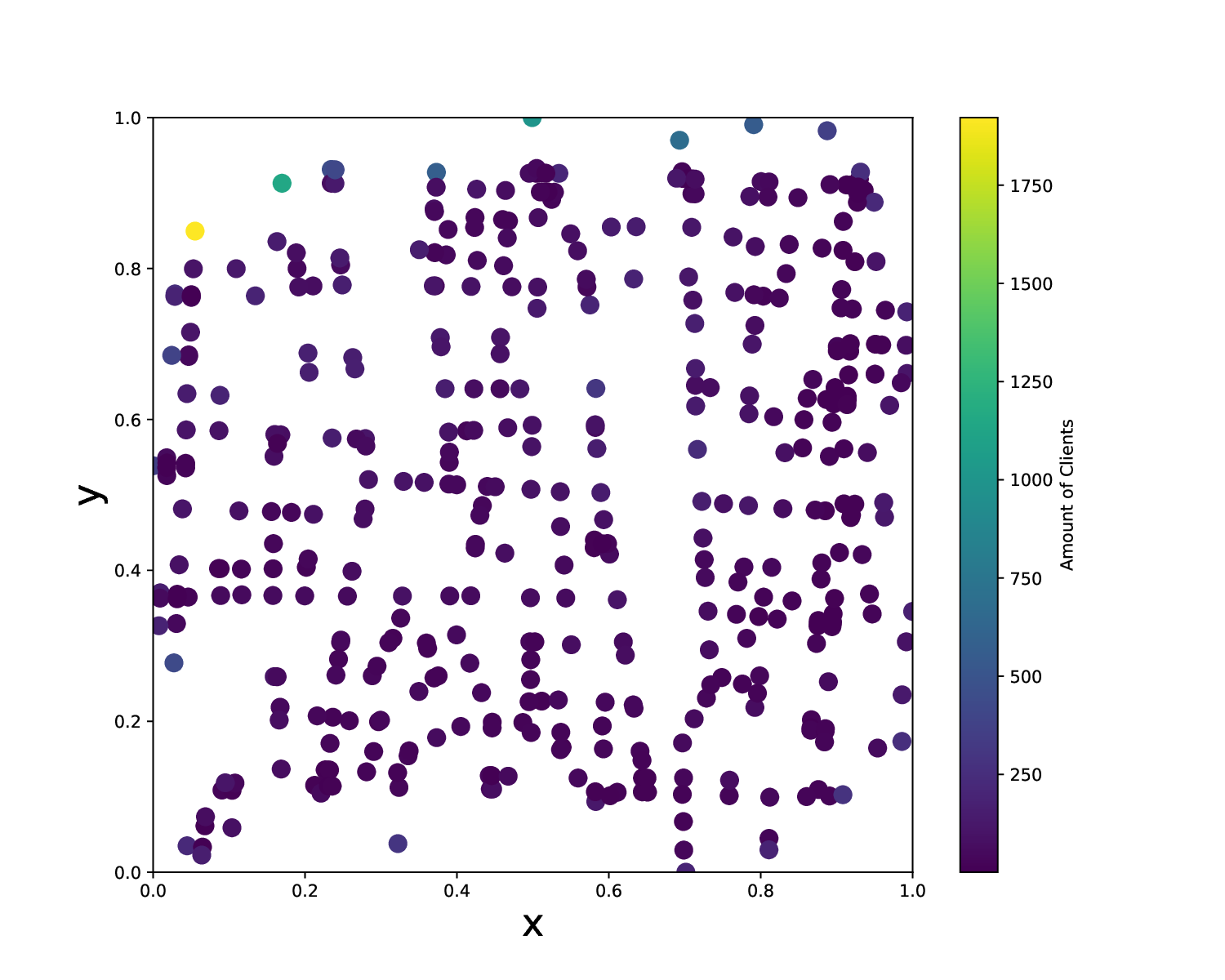}
    \caption{Chiang Mai instance normalized to a $1 \times 1$ window}
    \label{fig:real_world_instance}
\end{figure}
Over 100 instances, the private algorithms were executed with $\varepsilon = 0.1$ and $\alpha = 0.1$.

\paragraph{Our Results} Figure \ref{fig:real_world_plot} shows that for $\delta = 0.1$ the reconnection algorithm outperforms the straightforward approach.
With this we can conclude that by executing the reconnection algorithm with a correctly chosen $\delta$ our private reconnection algorithm outperforms the straightforward approach. 

\begin{figure}[h!]
    \centering
    \includegraphics[width=0.5\linewidth]{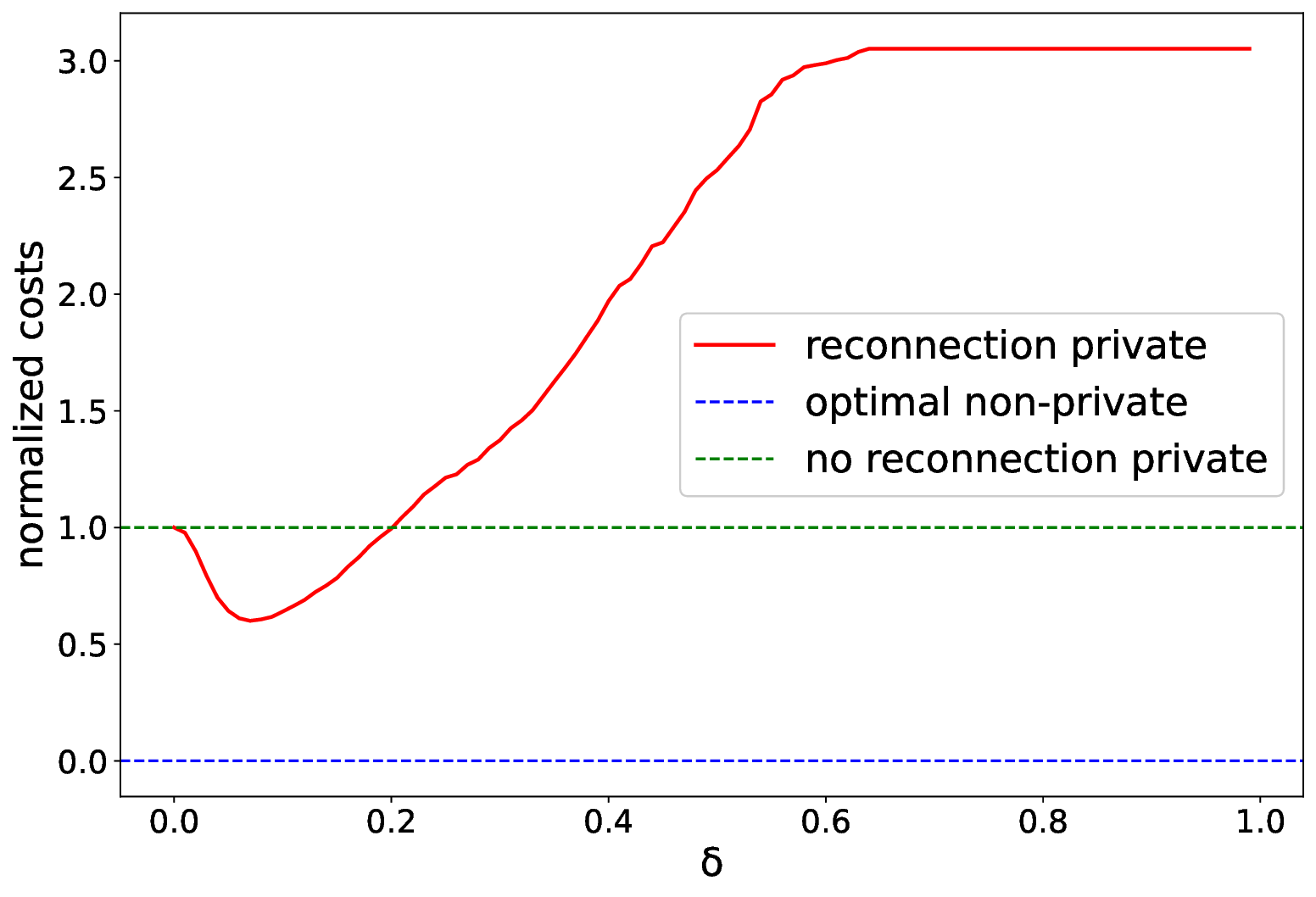}
    \caption{Normalized costs for varying \(\delta\) on real-world instances}
    \label{fig:real_world_plot}
\end{figure}

\section{Conclusion and Future Works}
Without the super-set assumption, releasing results for the facility location problem under LDP leads to large errors because the noise added to each location is significant relative to its original value. In this work, we introduce an algorithm called ``re-connection'' that aggregates several values. This approach ensures that the aggregated original value is sufficiently large, so that the relative impact of the noise is reduced. 
As a result, our algorithm achieves a constant approximation ratio with only a small additive error. In the next version we plan to include further experiments on real-world datasets.
We believe that this technique can be extended to other algorithms operating under LDP, and we are currently exploring its application to additional combinatorial optimization problems.

\begin{credits}
\subsubsection{\ackname} Quentin Hillebrand is partially supported by KAKENHI Grant 20H05965, and by JST SPRING Grant Number JPMJSP2108. Vorapong Suppakitpaisarn is partially supported by KAKENHI Grant 21H05845 and 23H04377.
The authors express their gratitude to Phapaengmuang Sukkasem and Suphanat Chaidee for providing the real-world dataset used in this study. This research was conducted while Kevin Pfisterer was at The University of Tokyo, and the authors also thank Hiroshi Imai and Kunihiko Sadakane for hosting him. 

\end{credits}

%
%
\bibliographystyle{splncs04}
\bibliography{references}

\end{document}